\documentclass[12pt]{article}
\usepackage{graphicx}
\usepackage{amsmath,amsthm,amssymb,mathrsfs,mathtools}
\usepackage[OT1]{fontenc} 
\usepackage{etoc}
\usepackage{float}
\usepackage{fancyhdr}

\def\subsectiontitle{}
\def\subsubsectiontitle{}
\usepackage[backgroundcolor=orange!50!white]{todonotes}
\usepackage{tikz}
\usetikzlibrary{automata, positioning}

\usepackage[noamsmath]{kpfonts}
\usepackage{inconsolata}

\usepackage{enumitem}
\usepackage[margin=1in]{geometry}
\usepackage[colorlinks,breaklinks=true]{hyperref}
\usepackage{breakcites}

\usepackage{xcolor}
\AtBeginDocument{%
	\hypersetup{%
    linkcolor={red!50!black},%
    citecolor={blue!50!black},%
    urlcolor={blue!80!black}%
}%
}%

\usepackage[nameinlink,noabbrev,sort,capitalise]{cleveref}
\makeatletter
\def\ps@pprintTitle{%
 \let\@oddhead\@empty
 \let\@evenhead\@empty
 \def\@oddfoot{\emph{Very preliminary version}\hfill\emph{This draft: \today}}%
 \let\@evenfoot\@oddfoot}
\makeatother
\usepackage{etoolbox}
\patchcmd{\pprintMaketitle}
  {\fi\hrule}
  {\fi\ifvoid\extrainfobox\else\unvbox\extrainfobox\par\vskip10pt\fi\hrule}
  {}{}

\newsavebox\extrainfobox

\newtheorem{prop}{Proposition}
\crefname{prop}{Proposition}{Propositions}

\newtheorem{thm}{Theorem}
\crefname{thm}{Theorem}{Theorems}

\newtheorem{cor}{Corollary}[thm]
\crefname{cor}{Corollary}{Corollaries}

\newtheorem{lem}{Lemma}
\crefname{lem}{Lemma}{Lemmas}

\newtheorem{ass}{Assumption}
\crefname{ass}{Assumption}{Assumptions}

\newtheorem{axiom}{Axiom}
\crefname{axiom}{Axiom}{Axioms}

\newtheorem{defi}{Definition}
\crefname{defi}{Definition}{Definitions}

\theoremstyle{remark}
\newtheorem*{remark}{Remark}

\theoremstyle{definition}

\crefname{eg}{Example}{Examples}

\crefname{problem}{Problem}{Problems}

\newcommand{\overbar}[1]{\mkern 1.5mu\overline{\mkern-1.5mu#1\mkern-1.5mu}\mkern 1.5mu}

\usepackage{indentfirst}

\allowdisplaybreaks

\let\oldfootnote\footnote
\renewcommand\footnote[1]{\oldfootnote{\hspace{.4mm}#1}}

\makeatletter
\renewenvironment{proof}[1][\proofname] {\par\pushQED{\qed}\normalfont\topsep6\p@\@plus6\p@\relax\trivlist\item[\hskip\labelsep\bfseries#1\@addpunct{.}]\ignorespaces}{\popQED\endtrivlist\@endpefalse}
\makeatother

\let\oldFootnote\footnote
\newcommand\nextToken\relax

\renewcommand\footnote[1]{%
    \oldFootnote{#1}\futurelet\nextToken\isFootnote}

\newcommand\isFootnote{%
    \ifx\footnote\nextToken\textsuperscript{,}\fi}

\usepackage{natbib}
\def\d{\mathrm{d}}

\def\E{\mathbb{E}}
\def\rep{\rightharpoonup}

\begin{document}

\title{The Indirect Cost of Information} 
\author{Weijie Zhong\thanks{Stanford University, email: \url{weijie.zhong@stanford.edu}. This paper subsumes two earlier versions of the working paper circulated with title ``Dynamic information acquisition with linear waiting cost'' and ``Indirect information measure and dynamic learning''.}}
\date{March 2020}
\maketitle

This is a preliminary draft and the work is still in progress. Please bear with numerous mistakes.

\begin{abstract} 
	We study the \emph{indirect cost of information} from sequential information cost minimization. A key sub-additivity condition, together with monotonicity equivalently characterizes the class of indirect cost functions generated from any direct information cost. Adding an extra (uniform) posterior separability condition equivalently characterizes the indirect cost generated from any direct cost favoring incremental evidences. We also provide the necessary and sufficient condition when prior independent direct cost generates posterior separable indirect cost.\\
\end{abstract}
	
\section{Introduction}

Information plays a central role in economic activities and it is often endogenously acquired by decision maker, as opposed to being exogenously endowed. Therefore, it is important to understand the trade-off between the value of information and the cost of acquiring information. The value of information is often unambiguous in a single agent decision problem with uncertainty. It is measured by the increased expected utility from choosing optimal actions measurable to the signal realizations(see \cite{blackwell1951comparison}). However, there has been less consensus on the proper form of information cost. On the one hand, the production of information is often very complicated and the corresponding cost is hard to measure. On the other hand, the technology used to produce information differs across different environments. So even if we can identify the cost of information in one setup, it is hard to provide general predictions.  \par

In this paper, we try to overcome these difficulties by studying the \emph{indirect cost of information}. We impose minimal assumptions on the \emph{direct cost of information} --- the actual cost at which an information structure is generated in a specific environment --- to obtain maximal robustness of the model. Then we obtain predictive power of the model through \textbf{optimality}: the indirect cost of an information structure is defined as the lowest total expected cost of acquiring a sequence of information structures which eventually replicate the target information structure. With this framework, we seek to answer two questions: What is the implication of sequential cost minimization without any prior knowledge about the direct cost? What restrictions are we imposing on the direct information cost when we assume the indirect cost to be (uniformly) posterior separable?\par

The first main result of the paper is a complete characterization of the class of indirect information cost functions. It is equivalently characterized by two key axioms: \cref{axiom:1} (monotonicity), Blackwell more informative information structure costs weakly more. \cref{axiom:2} (sub-additivity), if acquiring an information structure can be decomposed into two steps, then the cost of it as a whole is weakly lower than the total cost of the two steps. One implication of the main result is the optimality of Poisson type signal process in dynamic learning. If an impatient decision maker pays a convex transformation of an indirect information cost on the information acquired each day to implement a target information structure, then the optimal dynamic signal process resembles a compound Poisson process.\par

The second main result of the paper is a complete characterization of direct information cost functions which generate (uniformly) posterior separable indirect costs. An indirect information cost is (uniformly) posterior separable \emph{iff} it minimizes a direct information cost which \emph{favors incremental evidences} --- acquiring a signal process in the form of a Gaussian diffusion process yields total direct cost weakly lower than the direct cost of acquiring a single information structure which contains exactly the same amount of information. The first result of the paper reveals that sequential optimization generates sub-additivity in the indirect information cost. The second result illustrates that such sub-additivity is generically strict. When, and only when acquiring incremental evidences that shift belief only locally always costs less than other ways to sequentially acquire information, the ``super-additivity'' in the direct information cost is sufficiently strong to offset all the sub-additivity in the indirect cost generated by the optimization process.\par

Then we study the implications of our main results. In \cref{ssec:prior:independent}, we link ``prior-independence'' and posterior separability through optimization. These two properties, known to be incompatible on a single information cost, can be reconciled as a prior independent direct cost can lead to posterior separable indirect cost. We provide its necessary and sufficient condition: when there are two states and direct cost is bounded below by the standard mutual information plus an extra unweighted average of Kullback-Leibler divergence from the signal's marginal distribution to conditional distribution.

\subsection{Related Literature}
The question of how to measure the informativeness of an experiment has been extensively studied. On the one hand, people have proposed and utilized various functional forms to measure the amount of information. \cite{sims2003implications} introduced a rational inattention framework at the center of which is an Entropy based mutual information cost of information. \cite{caplin2013rational} studied the implication of mutual information cost and generalized it to the class of (uniformly) posterior separable costs. \cite{hebert2016rational} proposed a class of ``neighborhood-based'' information costs. The Kullback-Leiber divergence (\cite{kullback1951information}), or other forms of divergences have also been applied to measure the changes in uncertainty induced by information. On the other hand, a growing number of papers seek to provide axiomatic foundations for information costs. \cite{caplin2015revealed} provides an behavioral axiomatic foundation for general information acquisition representation. \cite{NBERw23652} then provides behavior axioms for (uniformly) posterior separable information cost and mutual information cost. \cite{frankel2018quantifying} and \cite{mensch2018cardinal} each provide different set of axioms on the measure of information that characterize uniform posterior separability. \cite{pomatto2018cost} characterizes the class of information costs satisfying a constant marginal cost axiom.\par

The motivation of this paper is in the same vein as the papering microfounding information cost functions. However, analysis of this paper fundamentally differs from the other axiomatic frameworks in its focus on optimality. In other words, we take sequential optimality of information gathering as a ``meta-axiom'', and study the implications of that. As a result, we can avoid strong regulations on the direct information cost (for example, allowing it to be ``prior independent'', which has to be excluded in other rational inattention based frameworks). Meanwhile, optimality also allows the indirect information cost to be sufficiently different from the direct cost so that we can accommodate and provide sharp characterizations for popular cost forms like mutual information and divergence based costs.\par

Conceptually, this paper builds on \cite{morris2017wald}, which shows that the total cost of Wald sampling using only Gaussian signals is uniformly posterior separable. The second result of the paper essentially provides the maximal generalization of it: on the one hand, the posterior separability of total cost holds even when all other sampling technologies are available, as long as each individual signal is more expensive comparing to its Gaussian replication.\footnote{This statement is actually proved to be a corollary of \cite{morris2017wald}'s main result.} One the other hand, we also show that the aforementioned condition is necessary for the total cost of optimal Wald sampling to be posterior separable.\par

\vspace{1em}

The rest of the paper is organized as follows. In \cref{sec:model}, we introduce the main model and the key axioms. \cref{sec:characterization} provides the two main characterization results. \cref{sec:implication} studies further implications of the main results.

\section{Model}
\label{sec:model}

In this section, we first introduce the direct and indirect cost of information. Then we introduce the two key axioms that will be used in next section to characterize indirect cost of information.

\subsection{The direct and indirect cost of information}

Let $X$ be a finite state space and $x\in X$ be an unknown state of the word. $\Delta(X)$ denotes the space of probability measures on $X$ ($\Delta(X)$ is a $(|X|-1)$-dimensional real simplex). $\Delta^2(X)$ denotes the space of probability measures on $\Delta(X)$, equipped with the L\'evy-Prokhorov metric.\footnote{$\Delta(X)$ is a $(|X|-1)$-dimensional real simplex. $\Delta^2(X)$ is a complete, separable and compact metric space.} Any $\pi\in \Delta^2(X)$ is called an \emph{information structure}. The standard way to define an information structure is through defining the joint distribution of the signal and the unknown state. Each such joint distribution induces a measure of posterior beliefs according to Bayes rule, which on expectation equals the prior. In this paper, we (mostly) take the belief-based approach, for its notational simplicity and to avoid defining a class of ``arbitrary'' signal alphabets. $\forall \pi\in \Delta^2(X)$, it defines both the prior belief $\E_{\pi}[\nu]$, and the induced measure of posterior beliefs.\par

An information structure in $\Delta^2(X)$ is called \emph{generic} in $X$ if the prior belief has full support on $X$. $\forall X'\subset X$, let $\mathcal{I}(X')\subset \Delta^2(X)$ denote the subset of all information structures generic in $X'$. \footnote{It is easy to verify that $I(X')$ is relatively open in $\Delta^2(X')$.} Then $\Delta^2(X)$ can be partitioned into $\left\{\mathcal{I}(X')\right\}_{X'\subset X}$. \par

Given any information structure $\pi$, the direct cost of acquiring $\pi$ is defined by:

\begin{defi}
	A \textbf{direct information cost} function $C$ maps $\Delta^2(X)$ into $\mathbb{R}^+$.
	\label{defi:direct}
\end{defi}
Unless stated otherwise, we consider direct information cost functions $C$ that are bounded on $\Delta^2(X)$, continuous on each generic subspace $\mathcal{I}(X')$ \footnote{We allow continuity to break when information structure becomes non-generic to allow possible ``prior independent'' information cost functions, whose cost might jump when the belief on a state diminishes. Since there are only finitely many $\mathcal{I}(X')$ and $C$ is bounded, this weaker continuity is sufficient to guarantee integrabiltiy of $C$ w.r.t. probability measures on $\Delta^2(X)$.} and satisfies the following axiom:

\setcounter{axiom}{-1}
\begin{axiom}\label{axiom:0}
	$\forall \mu\in \Delta(X)$: $C(\delta_{\mu})=0$; $\forall P(\pi)\in \Delta^3(X)$ s.t. $\E_{\pi}[\nu]\equiv \mu_0$ on $\text{supp} (P)$ , $C(\E_P[\pi])\le\E_{P}[C(\pi)]$.
\end{axiom}

\cref{defi:direct,axiom:0} can be interpreted as that the direct cost $C(\pi)$ is already an ``indirect cost'' of a static optimization problem, where 1) for all alphabets and signal structures inducing the same measure of posterior beliefs, the one with minimal cost is chosen, 2) a trivial information structure is replaced by ``do-nothing'' which costs zero and 3) one can take a mixed strategy and randomize between different information structures. Since such static optimization is quite simple to implement, we just assume that these properties are carried directly by a direct information cost. Let $\mathcal{C}$ denote the set of all direct information cost functions satisfying continuity and \cref{axiom:0}. Other than continuity and \cref{axiom:0}, direct information cost can be completely flexible. In particular, it can be restricted to depend on only the actual signal structure (defined by conditional distribution of signals) but not the prior belief --- capturing the idea of a physical cost.\par

Now we consider the sequential minimization of information cost. $\forall \pi\in \Delta^2(X)$, we first define the belief processes that replicate the information structure $\pi$. 

\begin{defi}\label{defi:rep}
	$\forall \pi\in \Delta^2(X)$, A $2T$-period Markov chain $\langle \mu_t \rangle$ (define on $(\Omega,\mathcal{F},\mathcal{P})$) \textbf{replicates} $\pi$ if 1) $\mu_{2T}\sim \pi$, 2) $\E[\mu_{2t+1}|\mu_{2t}]=\mu_{2t}$ and 3) $\E[\mu_{2t-1}|\mu_{2t}]=\mu_{2t}$. 
\end{defi}

The first condition means $\langle \mu_t \rangle$ eventually replicates $\pi$. The second condition means from any period $2t$ to $2t+1$, information is acquired and belief is updated according to Bayes rule. The third condition means from any period $2t-1$ to $2t$, information is ``discarded'' and belief contracts. By defining $\langle \mu_t \rangle$ as in \cref{defi:rep}, it is implicitly assumed that 1) acquiring information only measurable to belief and time is sufficient (for optimality, which will be defined and proved later) and 2) information is freely disposable. The conditions in \cref{defi:rep} are denoted by $\langle \mu_t \rangle\rep \pi$. \par

The we define the indirect cost of information:

\begin{defi}
		\label{defi:indirect}
		$C^*:\Delta^2(X)\to \mathbb{R}^+$ is an \textbf{indirect information cost} function if $\exists$ direct information cost $C:\Delta^2(X)\to \mathbb{R}^+$ s.t $\forall \pi\in \Delta^2(X)$:
	\begin{align}
		C^*(\pi)= \inf_{\langle \mu_t \rangle}\ & \E\left[ \sum_{t=0}^{T-1}C(\pi_{2t}(\mu_{2t+1}|\mu_{2t})) \right] \label{eqn:1}\\
		\text{s.t.}\ &\langle \mu_t \rangle \rep \pi\notag
	\end{align}
	where $\pi_{2t}(\mu_{2t+1}|\mu_{2t})$ denotes the conditional distribution of $\mu_{2t+1}$ on $\mu_{2t}$.
\end{defi}

\cref{eqn:1} defines a program that searches for a cost minimizing belief process that replicates the target information structure $\pi$. In the objective function, only in even periods the cost of belief change is counted because by definition in even periods information is acquired and in odd periods information is freely discarded. The optimization over $\langle \mu_t \rangle$ implicitly allows $T$ to be chosen as well. The integrability of \cref{eqn:1} is in general not guaranteed, but for $C\in \mathcal{C}$ the expression is well defined. $\forall C\in \mathcal{C}$, \cref{eqn:1} is a well defined non-negative real number, and hence a map $\phi$ can be defined as $\phi(C)=C^*$ indicating that $\phi(C)$ is the indirect information cost derived from solving \cref{eqn:1} with direct information cost $C$. \par

In the appendix, I prove \cref{lem:markov} which shows that we can consider a more complicated sequential signal structure which is not necessarily Markovian. However, the terminal belief distribution can always be replicated by a process satisfying \cref{defi:rep} and with weakly lower total cost. Therefore, \cref{eqn:1} could be thought as the maximal flexibility benchmark.\par

We can partially generalize our main results to optimization problems less flexible than \cref{eqn:1}. First, the main characterization theorem (\cref{thm:1}) can also be established without free disposal of information, in which case $\langle \mu_t \rangle$ has to be a martingale. Then this restriction exactly takes away one axiom in the characterization. Second, the characterization for direct costs which generate uniformly posterior separable indirect cost (\cref{thm:ps}) generalizes to any map $\widetilde{\phi}$ between $\phi$ and the identity map. In other words, the characterization generalizes when arbitrary extra restrictions are placed on the optimization problem \eqref{eqn:1}.

\subsection{Axioms}
In this subsection, we present the two key axioms that characterize $\phi(\mathcal{C})$. 

\begin{axiom}[Blackwell monotonicity]\label{axiom:1}
	$\forall \pi,\pi'\in \Delta^2(X)$ and $\pi\le_{MPS}\pi'$: $C(\pi)\le C(\pi')$.
\end{axiom}

\cref{axiom:1} states that information cost $C$ is consistent with the Blackwell order (\cite{blackwell1951comparison}). A more informative information structure costs weakly more to acquire. It is straight forward that if we allow free disposal of information, then \cref{axiom:1} trivially holds.

\begin{axiom}[Subadditivity]\label{axiom:2}
	$\forall \pi'(\mu,\nu) \in \Delta(\Delta(X)^2)$, with marginal distribution $\pi(\mu)$ and $\pi''(\nu)$: 
	\begin{align*}
	C(\pi'')\le C(\pi)+\E_{\pi}[C(\pi'(\cdot |\mu))]
\end{align*}
\end{axiom}

\cref{axiom:2} is a sub-additivity condition. It states that if acquiring an information structure can be decomposed into two steps, then the cost of the information structure should be weakly lower than the total expected cost of the two steps. \par

The first axiom is a quite standard one, satisfied in most existing reduced-form information cost representations. The second axiom is less standard. The most well-known information cost that satisfies \cref{axiom:2} is the Entropy based mutual information cost (\cite{sims2003implications}) and the more general uniformly posterior separable information cost (\cite{caplin2013rational}). In fact, it will be shown in \cref{lem:add} that uniform posterior separability is equivalent to a stronger \emph{additivity}. In practice, it might not be easy to verify whether an information cost function satisfies \cref{axiom:2}. In this paper we do not view the axioms as falsifiable behavior predictions as in most decision theory studies. Instead, we take them as pure mathematical axioms that emerges from the ``meta-axiom'' of optimality and study their implications.

\section{Characterization}
\label{sec:characterization}
In this section, we present the two main characterization theorems. The first theorem characterizes the class of indirect information costs derived from any direct information cost. The second theorem characterizes the class of direct information costs which generate (uniformly) posterior separable indirect information costs.

\begin{thm}\label{thm:1}
	$\phi(\mathcal{C})=\left\{ C|C\in\mathcal{C} \text{ satisfying \cref{axiom:1,axiom:2}} \right\}$ and
		$\phi(C)
		\begin{cases}
			= C&\forall C\in \phi(\mathcal{C})\\
			< C&\forall C\not\in \phi(\mathcal{C})
		\end{cases}$.\footnote{$<$ is defined as the standard strict (not strong) partial order on function space.}
\end{thm}

 \cref{thm:1} states that the class of monotonic and subadditive information costs is exactly the class of indirect information costs. It is quite straight forward that any monotonic and subadditive information cost $C$ is an indirect information cost. If we assume this cost $C$ itself being the direct information cost, then $C$ is unimprovable in \cref{eqn:1}: Take any $\langle \mu_t \rangle\rep \pi$, 1) it is always weakly improving not to discard any information at $t=1$ by \cref{axiom:1}, and hence $\mu_1=\mu_2$. 2) it is always weakly improving combining the information structure acquired at $t=0$ and the information structure acquired at $t=2$ by \cref{axiom:2}. Then recursively apply this argument to the finite process $\langle \mu_t \rangle$ and we can conclude that it is weakly better to acquire $\pi$ directly. Therefore, $C=\phi(C)$ and this also implies the second part of the result.\par

 The opposite direction is also very intuitively. If we can replicate each of $\pi$ and $\pi'(\nu|\mu)$ with sequential processes, then take the sequence replicating $\pi$ and append each sequence replicating $\pi'$ after that should replicate $\pi''$. However, the proof is slightly more involved than the intuition because we can not guarantee that the so constructed sequence has finite horizon.\footnote{Modifying the definition of replicating process to be infinite horizon creates more problems because it is not well-defined to append a process ``after'' an infinite-horizon process.} Besides, \cref{thm:1} also claims that $\phi$ is injective on $\mathcal{C}$, namely \cref{axiom:0} and continuity are preserved through sequential optimization. \par

 \cref{thm:1} implies that if a decision maker is facing an indirect information cost, then acquiring a large chunk of information should be weakly cheaper than acquiring small pieces and combine them together. However, in many realistic decision problems, acquiring a lot of information instantaneously might not be feasible/efficient. When a decision maker decomposes an information structure into asymptotically zero-cost pieces (measured by an indirect cost function), the following lemma illustrates the cheapest way to do it.

 \begin{lem}\label{lem:pos}
	$\forall C$ satisfying \cref{axiom:0,axiom:2}, $\forall \pi\in \Delta^2(X)$ and $\mu_0=\E_{\pi}[\nu]$, define $\pi_{\lambda}=\lambda\pi+(1-\lambda)\delta_{\mu_0}$, then $C(\pi_{\lambda})=\lambda C(\pi)$.
\end{lem}

\begin{proof}
	First, \cref{axiom:0} direct implies that $C(\pi_{\lambda})\le \lambda C(\pi)$. Second, define $\pi'(\nu|\mu_0)=\pi(\nu)$ and $\pi'(\nu|\mu)=\delta_{\mu}$ if $\mu\neq \mu_0$. Then $\E_{\pi}[\pi'(\nu|\mu)]=\pi$. \cref{axiom:0,axiom:2} implies that $C(\pi)\le C(\pi_{\lambda})+(1-\lambda)C(\pi)\implies C(\pi_{\lambda})\ge C(\pi)$.
\end{proof}

The property stated in \cref{lem:pos} is the third axiom in \cite{pomatto2018cost}, in which $\pi_{\lambda}$ is called an dilution of $\pi$. $\pi$ can be (approximately) replicated using $\pi_{\lambda}$ by running $\pi_{\lambda}$ until belief jumps away from prior $\mu_0$, and the total cost of the process is exactly $\frac{1}{\lambda}C(\pi_{\lambda})$. So \cref{lem:pos} is essentially stating that an indirect information cost is additive when an information structure is decomposed into its dilutions. Since cost is subadditive in all other kinds of decompositions, this implies that among all decompositions that incur cost $\lambda C(\pi)$ per period, the dilution $\pi_{\lambda}$ incurs the lowest total cost. In the continuous-time limit when the flow cost is converging to zero, the optimal decomposition resembles a compound Poisson process. With small probability, an information signal arrives and belief jumps according to $\pi$. Otherwise, belief stays at the prior.

\subsection{(Uniform) posterior-separability}

In this subsection, we restrict our attention to indirect information costs which are (uniformly) posterior separable. We adopt the notion of posterior-separable information cost and uniformly posterior-separable information cost from \cite{caplin2013rational} and \cite{NBERw23652}. $C$ is called \emph{posterior separable} if there exists non-negative (divergence) function $D(\nu||\mu)$ s.t. $C(\pi)=\E_{\pi}[D(\nu||\E_{\pi}[\nu])]$.\footnote{Within the scope of this paper, we only consider divergence $D(\nu||\mu)$ continuously differentiable in $\mu$ and twice differentiable in $\nu$ for technical convenience.} $C$ is called \emph{uniformly posterior separable} if there exists convex (potential) function $H(\mu)$ s.t. $C(\pi)=\E_{\pi}[H(\nu)]-H(\E_{\pi}[\nu])$.

\begin{lem}\label{lem:add}
	$C\in \mathcal{C}$ is uniformly posterior separable \emph{iff} $C$ satisfies \cref{axiom:1} and \emph{additivity}.
\end{lem}

\begin{proof}
	Uniform posterior separability $\implies$ additivity is trivial. Now we show the converse. $\forall \mu$, let $\overbar{\pi}_{\mu}$ be the fully revealing information structure whose support is $\left\{ \delta_{x} \right\}$. Define potential function $H(\mu)=-C(\overbar{\pi}_{\mu})$. Define $\widehat{C}(\pi)=\E_{\pi}[H(\nu)]-H(\E_{\pi}[\nu])$. Obviously, $\widehat{C}$ is uniformly posterior separable. $\widehat{C}(\pi)=C(\overbar{\pi}_{\E_{\pi}[\nu]})-\E_{\pi}[C(\overbar{\pi}_{\nu})]=C(\pi)$, because of $\E_{\pi}[\overbar{\pi}_{\nu}]= \overbar{\pi}_{\E_{\pi}[\nu]}$ and additivity. Therefore, $C\equiv\widehat{C}$.
\end{proof}

Given \cref{lem:add}, it is clear that the class of indirect information costs $\phi(\mathcal{C})$ and the class of uniformly posterior separable cost functions differ by exactly a ``super-additivity'' condition. In this section, we study what conditions on direct cost functions $\mathcal{C}$ give us this extra super-additivity on the indirect cost function.\par

To state the additional conditions on $\mathcal{C}$, we impose an extra technical assumption on the direct cost $C$, which allows us to approximate the information cost function using second-order Taylor expansion locally when the information structure is approximately uninformative.

\begin{ass}\label{ass:diff}
	There exists continuous symmetric matrix valued function $B(\mu)$ s.t. $\forall \mu_0\in \Delta(X)^o$, $\forall \epsilon$, $\exists \delta$ s.t. $\forall \pi\in \Delta^2(X)$ s.t. supp$(\pi)\subset B_{\delta}(\mu_0)$, let $\mu=\E_{\pi}[\nu]$:
	\begin{align}
		\left|C(\pi)-\E_{\pi}\left[ (\nu-\mu)B(\mu)(\nu-\mu) \right]\right|\le \epsilon\cdot\E_{\pi}[\|\nu-\mu\|^2]
		\label{eqn:diff}
	\end{align}
	where $B(\mu)$ is positive semi-definite*.\footnote{\label{fn:8}Defined as the bilinear form $y^TB(\mu)y\ge0$ for any $y\cdot \mathbf{1}=0$. This is equivalent to requiring $\widetilde{B}(\mu)=[I,-\mathbf{1}]\cdot B(\mu)\cdot[I,-\mathbf{1}]^T$ to be PSD. Notice that since $\nu-\mu$ always adds up to $0$, we are essentially considering a $|X|-1$-dimensional subspace, in which the bilinear form $\widetilde{B}(\mu)$ is uniquely pinned down. $B(\mu)$ is unique only up to a $f(\mu)\mathbf{1}\cdot\mathbf{1}^T$ term. } We say $B(\mu)$ locally characterizes $C(\pi)$. 
\end{ass}

\cref{ass:diff} states that when an information structure only shifts belief locally, then the cost of the information structure is approximated by the expectation of a bilinear form, which calculates the squared generalized distance of belief oscillation from prior belief to each posterior belief. \cref{ass:diff} might seem restrictive because i) there is no zero and first order term, and ii) the second order term is linearly separable in posterior beliefs. In later discussions, we show in \cref{lem:diff} that \cref{eqn:diff} is the consequence of a seemingly much weaker (local) twice differentiability w.r.t. the conditional signal distribution.

\begin{defi}\label{defi:FIE}
	$\forall C\in\mathcal{C}$ satisfying \cref{ass:1}, let $B(\mu)$ characterizes $C(\pi)$ locally. $C$ \textbf{favors incremental evidences} if 1) $B(\mu)$ is a Hessian matrix$^*$\footnote{Like in \cref{fn:8}, $\widetilde{B}(\mu)$, which is $B(\mu)$'s projection to $\mathbb{R}^{|X|-1}$, is a Hessian matrix of function $H$'s projection to $\mathbb{R}^{|X|-1}$. } of $H:\Delta(X)\mapsto \mathbb{R}$; and 2) $\forall \pi\in \Delta^2(X)$, $$C(\pi)\ge \E_{\pi}\left[ H(\nu) \right]-H(\E_{\pi}[\nu])$$
\end{defi}

We name the condition in \cref{defi:FIE} \emph{favoring incremental evidences} because of the following connection discovered by \cite{morris2017wald}: When $|X|=2$, any information structure can be replicated by a stopped Gaussian process with flow variance $\sigma^2(\mu)$. Suppose the direct cost of acquiring one more unit time of signals is $b(\mu)\sigma^2(\mu)$, then the total cost to replication $\pi$ is exactly $\E_{\pi}[H(\nu)]-H(\mu)$, where $H''(\mu)=b(\mu)$. \footnote{\cite{morris2017wald} uses a different notation. If we define $c(\mu)=b(\mu)\sigma^2(\mu)$ and $\phi(\mu)=H(\mu)$, then the statement is exactly Theorem 1 of \cite{morris2017wald}.} A reasonable conjecture is that if information is only acquired in the form of incremental evidences, which induces belief diffusion, and $B(\mu)$ is a Hessian matrix of function $H(\mu)$, then $\underline{C}(\pi)=\E_{\pi}[H(\nu)]-H(\mu)$ is the total cost to replicate the information structure $\pi$. Then naturally, $C$ ``favors incremental evidences'' if $\underline{C}$ is below $C$.\par

The previous conjecture, which is an extension of \cite{morris2017wald}, almost implies that if $C$ favors incremental evidences then $\phi(C)$ is $\underline{C}$ itself, which is uniformly posterior separable. In \cref{thm:ps}, we will show that not only this conjecture is true, but its converse is also true. 

\begin{ass}\label{ass:2}
	$\exists m>0$ s.t. $\forall \pi\in \Delta^2(X)$, $C(\pi)\ge m\cdot \E_{\pi}\left[ \|\nu-\E_{\pi}[\nu]\|^2 \right]$.
\end{ass}

\begin{thm}\label{thm:ps}
	$\forall C\in\mathcal{C}$ satisfying \cref{axiom:1,ass:diff,ass:2}: \center{$\phi(C)$ is (uniformly) posterior separable  $\iff$ $C$ favors incremental evidences}
\end{thm}

\cref{thm:ps} provides the complete characterization of indirect information cost functions that satisfies the (uniform) posterior separability condition. Such indirect information must be derived through sequentially minimizing a direct information cost which favors learning from incremental evidences. In other words, assuming learning from incremental evidences to be cheaper adds exactly the right amount of super-additivity to the indirect information cost. Notice that \cref{thm:ps} also shows that the optimization process wipes out all the differences between posterior separability and uniform posterior separability. In other words, a posterior separable indirect information cost is uniformly posterior separable.\par

The sufficiency of \cref{defi:FIE} is a generalization of \cite{morris2017wald}, which is already explained. The necessity is the non-trivial part of the theorem. The key step to prove necessity is to show that although sequential optimization in general reduces the costs of information structures, the cost of Gaussian-diffusion-like information structures that shift belief only locally remain invariant (in the asymptotic sense). When $\phi(C)$ is (uniformly) posterior separable, then cost for these asymptotically uninformative information structures becomes sufficient to characterize the indirect cost of all other information structures. \par

The previous discussion about necessity suggests that \cref{defi:FIE} is actually much more powerful than stated in \cref{thm:ps}, because \cref{thm:ps} only utilizes the invariance property of $\phi$ but not its fine details. In fact, $\forall \phi'$ between the identity map and $\phi$, $\phi'$ also preserves the cost for asymptotically uninformative information structures by squeeze theorem. Thus a direct corollary of \cref{thm:ps} is that:
\begin{cor}\label{cor:ps}
	$\forall C\in\mathcal{C}$ satisfying \cref{axiom:1,ass:diff,ass:2}, $\forall \phi\le \phi' \le \mathrm{Id}$:
	\begin{align*}
		\phi'(C) \text{ is uniformly posterior separable }\implies \text{ $C$ favors incremental evidences}
	\end{align*}
\end{cor}
\cref{cor:ps} states that the condition in \cref{defi:FIE} is necessary for indirect cost to be uniformly posterior separable no matter what is the actual optimization problem. As long as the decision maker can at least passively pay the direct cost and at most do as good as in \cref{eqn:1}, posterior separable indirect cost implies direct cost favoring incremental evidences.

\subsubsection*{Differentiability of $C(\pi)$}

\cref{ass:diff} seems to put a lot restrictions on the structure of the derivative of $C(\pi)$ w.r.t. $\pi$ when $\pi$ is close to being uninformative. We now show that \cref{ass:diff} is in fact an implication of a seemingly much weaker twice differentiability.

To formally state the differentiability, we first reparametrize the space of information structures using the conditional signal distributions. Define $\widehat{C}_n(P,\mu):\Delta(n)^{|X|}\times \Delta(X)\to \mathbb{R}^+$. $P$ is a Markov transition matrix which represents an $n$-signal information structure of $X$. If $\pi$ is the corresponding (finite) distribution of posteriors, then $\widehat{C}_n(P,\mu)$ is defined to be equal to $C(\pi)$. $\widehat{C}_n$ is essentially a different parametrization of $C$ when the support of posterior distribution is finite. The differentiability assumption is imposed on $\left\{  \widehat{C}_n\right\}$:

\begin{ass}\label{ass:1}
	$\forall \mu_0\in \Delta(X)^O$, $\forall \epsilon>0$, $\exists \delta$ s.t. $\forall n,\forall p\in \Delta(n)$, $\forall \mu\in B_{\delta}(\mu_0)$, $\forall P\in \Delta(n)^{|X|}$ s.t. $\sum p_i \| \frac{P_i}{p_i}-\mathbf{1} \|^2\le\delta$ and $\max\|\frac{P_i}{p_i}-\mathbf{1}\|^2\le\delta$, $\exists !$ real function $\Lambda_n(p,\mu)$, $|X|\times n$-dimensional vector valued function $\Theta_n(q,\mu)$ and $|X|\times n$-dimensional symmetric matrix valued function $Q_n(q,\mu)$:
	\begin{align*}
		\left| \widehat{C}_n(P,\mu)-\left[ \Lambda_n(p,\mu)+\Theta_n(p,\mu)\cdot(P-p\cdot \mathbf{1}^T)+(P^T-\mathbf{1}\cdot p^T)\cdot Q_n(p,\mu)\cdot (P-p\cdot \mathbf{1}^T) \right]  \right|\le \epsilon\sum p_i\left\|\frac{P_i}{p_i}-\mathbf{1}\right\|^2
	\end{align*}
	where $\Lambda_n(p,\mu), \Theta_n(p,\mu)$ and $Q_n(p,\mu)$ are continuous in $p$ and $\mu$.\footnote{Note that although $P$ is usually interpreted as a Markov transition matrix, it is treated as a $|X|\times n$-dimensional vector in the expression.}
\end{ass}

\cref{ass:1} states that when the informativeness of an information structure $\pi$ (measured by $\sum p_i\|P_i/p_i-1\|^2$) is approximately zero, and each signal provides only incremental evidence (likelihood ratios are uniformly approaching zero), the information cost $\widehat{C}$ can be approximated quadratically up to error lower than the order of $\sum p_i\|P_i/p_i-1\|^2$. Notice that this differentiability requirement is equivalent to the standard definition, because for any fixed $n$ the norms on $\mathbb{R}^{n\times|X|}$ are all equivalent. However, when varying $n$, our condition implicitly requires that when $n$ increases the approximation error does not scale up. Other than the differentiability of each $\widehat{C}_n(\cdot,\mu)$, \cref{ass:1} does not put any regulation on how the quadratic approximation for different $n$ and $\mu$ are related. Given \cref{ass:1}, we can obtain a local approximation of $C$:

\begin{lem}\label{lem:diff}
	$\forall$ $C\in\mathcal{C}$ satisfies \cref{ass:1} there exists continuous symmetric matrix valued function $B(\mu)$ s.t. $\forall \mu_0\in \Delta(X)^o$, $\forall \epsilon$, $\exists \delta$ s.t. $\forall \pi\in \Delta^2(X)$ s.t. $\text{supp}(\pi)\subset B_\delta(\mu_0)$, let $\mu=\E_{\pi}[\nu]$,
	\begin{align*}
		\left| C(\pi)-\E_{\pi}\left[ (\nu-\mu)^TB(\mu)(\nu-\mu) \right] \right|\le \epsilon \E_{\pi}\left[\|\nu-\mu\|^2\right]
	\end{align*}
	where $B(\mu)$ satisfies $\forall \nu,\ \nu^TB(\mu)\nu\ge\mu^TB(\mu)\mu$.  
\end{lem}

\cref{lem:diff} states that under \cref{ass:1}, we can approximate $C(\pi)$ locally when the support of $\pi$ shrinks. Moreover, the approximating is characterized by the expectation of an inner product defined on vector space $\Delta(X)$, as is stated in \cref{ass:diff}. The cost of an information structure is essentially measured by the expected ``square-distance'' between prior and posterior beliefs. The distance is induced by the inner product defined by kernel $B(\mu)$. Note that the dimensionality of $(\Lambda_n,\Theta_n,Q_n)$ corresponding to an arbitrary function $\widehat{C}_n$ can be much larger than the dimensionality of $B(\mu)$. What \cref{lem:diff} proves is that when $\widehat{C}_n$ is an information cost function, the intrinsic invariance property of information cost (when difference $P$'s correspond to the same $\pi$) significantly reduces the degree of freedom in $(\Lambda_n,\Theta_n,Q_n)$ so that they can be represented by a much lower dimensional object $B(\mu)$. \par

\begin{remark}
	\cref{lem:diff} shows that \cref{ass:1} implies \cref{ass:diff}, but the converse is not necessarily true. The condition $\forall \mu,\nu,\ \nu^TB(\mu)\nu\ge\mu^TB(\mu)\mu$ is stronger than PSD$^*$ of $B(\mu)$. This is because the differentiability we defined in \cref{ass:1} is in fact stronger than \cref{ass:diff}. $\forall p\in \Delta(n)$, we required $\widehat{C}_n(P,\mu)$ to be differentiable for any $P\in \Delta(n)^{|X|}$ close to $p$. Meanwhile, in \cref{ass:diff}, we only consider $P$ in the subspace $T_{p,\mu}=\{P|P\cdot \mu =p\}$, since these are the only Markov transition matrices consistent with $\mu$ and $p$. In principle, we can use this stronger differentiability condition to derive \cref{ass:diff}. However, by doing this we are assuming away potentially interesting information cost functions e.g. the variance of belief movement.
\end{remark}

\section{Implications of indirect information cost}
\label{sec:implication}
\subsection{Mutual information}
\label{ssec:mutial:information}
A straight forward implication of \cref{thm:ps} is a characterization of mutual information based information cost. Let $P_{s,x}$ be the joint distribution of state $x$ and signal $x$, and $P_s$, $P_x$ be the corresponding marginal distribution. Then the mutual information between state and signal is:
\begin{align*}
	I(S;X)=&D_{KL}(P_{s,x}||P_s\otimes P_x)\\
	=&\E_{P_s}\left[ \E_{P_{x|s}}[\log(P_{x|s})] \right]-\E_{P_x}[\log(P_{x})]\\
	=&\E_{P_s}\left[ \sum \nu_s(x)\log(\nu_s(x)) \right]-\sum \mu(x)\log(\mu(x))
\end{align*}
where $\nu_s(x)$ is the posterior belief of $x$ conditional on signal $x$ and $\mu(x)$ is the prior belief of $x$. The last equality represents $I(S;X)$ as a uniformly posterior separable cost function. It is easy to verify that the normalized local characterization of cost function $I(S;X)$ is the Fisher information matrix:
\begin{align*}
	B(\mu)=\text{diag}(\mu)^{-1}-\mathbf{1}\cdot\mathbf{1}^{T}
\end{align*}
Therefore, \cref{thm:ps} implies that $\phi(C)$ is mutual information if and only if $C$ is locally characterized by the Fisher information matrix and $\forall \pi\in \mathcal{I}(X)$, $C(\pi)$ is weakly higher than the mutual information of $\pi$.

\subsection{Prior independent direct cost}
\label{ssec:prior:independent}
In this subsection, we study the implication of restricting our attention to ``prior independent'' direct information cost. In particular, we provide a complete characterization for such direct information costs which generate (uniformly) posterior separable indirect cost. \par

It seems a reasonable restriction that the direct cost of generating information should depend on some physical process/device that does not change with the prior belief. This ``prior independence'' should be captured by direct information cost functions that depend on and only on the conditional distribution of signals on each state. In other words, a necessary condition for prior independent direct cost is that $\widehat{C}_n(P,\mu)\equiv \widehat{C}_n(P)$. Notice that by \cref{lem:diff}, $\left\{ \widehat{C}_n(P,\mu) \right\}$ is sufficient for determining the cost function's local behavior and the corresponding criteria in \cref{defi:FIE}. So it is wlog to define prior independence only using $\left\{ \widehat{C}_n \right\}$ with finite signal space within the scope of this paper.

\begin{defi}\label{defi:prior:independent}
	Direct cost $C\in\mathcal{C}$ with corresponding $\left\{ \widehat{C}_n(P,\mu) \right\}$ is \textbf{prior independent} if there exists $\left\{ \widehat{C}_n(P) \right\}_{n\in\mathbb{N}}$ s.t. $\forall n\in\mathbb{N},\mu\in \Delta(X),P\in \Delta(n)^{X}$, $\widehat{C}_n(P,\mu)\equiv \widehat{C}_n(P)$ .
\end{defi}

An immediate implication of prior independent direct cost $C$ is that the local representation $B(\mu)$ takes a specific form. By the proof of \cref{lem:diff}, $B(\mu)=\text{diag}(\mu)^{-1}Q(\mu) \text{diag}(\mu)^{-1}$, where $Q(\mu)$ is a sub-matrix of $Q_n(p,\mu)$, hence is $\mu$ independent when $C$ is prior independent. Therefore, in this subsection, we consider local representation $B(\mu)$ in the form of:
\begin{align*}
	B(\mu)=\text{diag}(\mu)^{-1}\cdot Q\cdot\text{diag}(\mu)^{-1}
\end{align*}
where $Q$ is PSD and satisfies $Q\cdot\mathbf{1}=\mathbf{0}$.

\subsubsection{Binary states}

When the state is binary, i.e. $|X|=2$, we can explicitly write down the criteria in \cref{defi:FIE}. The condition  $Q\cdot \mathbf{1}=\mathbf{0}$ implies $Q_{11}+Q_{12}=0$, $Q_{12}+Q_{22}=0$. In other words:
\begin{align*}
	&Q=
	\begin{bmatrix}
		\alpha&-\alpha\\
		-\alpha&\alpha
	\end{bmatrix}\\
	\implies &B(\mu_1,\mu_2)=
	\alpha\times \begin{bmatrix}
		\frac{1}{\mu_1^2}&-\frac{1}{\mu_1\mu_2}\\
		-\frac{1}{\mu_1\mu_2}&\frac{1}{\mu_2^2}
	\end{bmatrix}
\end{align*}

We parametrize $(\mu_1,\mu_2)=(\mu,1-\mu)$, let $\widetilde{B}(\mu)$ define the 2-differential form in the space $[0,1]$. Then:
\begin{align*}
	\widetilde{B}(\mu)=&\left( \frac{\d (\mu_1,\mu_2) }{\d \mu} \right)B(\mu_1,\mu_2)\left( \frac{\d (\mu_1,\mu_2)^T }{\d \mu} \right)\\
	=&[1\ -1]
	\begin{bmatrix}
		\frac{\alpha}{\mu^2}&\frac{-\alpha}{\mu(1-\mu)}\\
		\frac{-\alpha}{\mu(1-\mu)}&\frac{\alpha}{(1-\mu)^2}
		\end{bmatrix}	\begin{bmatrix} 1\\-1\end{bmatrix}\\
		=&\frac{\alpha}{\mu^2(1-\mu)^2}
\end{align*}
It is easy to verify that $\widetilde{B}(\mu)$ is the second derivative of:
\begin{align}
	H(\mu)=\alpha\left( 2(\mu\log(\mu)+(1-\mu)\log(1-\mu))-\log(\mu)-\log(1-\mu) \right) \label{eqn:H:binary}
\end{align}
where $H$ is essentially unique module a linear function. $\forall P\in \Delta(n)^{X}$ and $\mu\in \Delta(X)^o$, let $P_s^x$ denote the $(s,x)$ element in the matrix $P$.  Define distribution $P_{s,x}(s,x)=\mu(x)\cdot P_s^x$, $P_s(s)=\sum_x \mu(x)\cdot P_s^x$. Then the condition in \cref{defi:FIE} can be written as $\widehat{C}_n(P)\ge \underline{\widehat{C}}_n(P,\mu)$, where 
\begin{align}
	\underline{\widehat{C}}_n(P,\mu)=2\alpha\left( D_{KL}(P_{s,x}||P_s\otimes \mu) +\frac{1}{2}\sum_x D_{KL}(P_s||P^x)\right)\label{eqn:ps:binary}
\end{align}
\cref{eqn:ps:binary} defines the lower bound for any information structure's direct cost such that the indirect cost of information is uniformly posterior separable. When \cref{eqn:ps:binary} is satisfied for all information structures, the uniformly posterior separable indirect information cost $\phi(C)$ is exactly defined by $\phi(C)(\pi)=\E_{\pi}[H(\nu)]-H(\E_{\pi}[\nu])$. \cref{eqn:ps:binary} includes two terms, the first term is exactly the \emph{mutual information} between the state and the signal. The second term is the unweighted average Kullback-Leibler divergence from the signal's marginal distribution to the conditional distribution. 

\subsubsection{More states}
When $|X|=K>2$, we still parametrize $\Delta(X)$ using the first $K-1$ entries of a belief distribution: $\mu=(\mu_1,\mu_2,\cdots,\mu_{K-1},1-\sum_{k=1}^{K-1} \mu_j)$. Let $\widetilde{B}(\mu_1,\cdots,\mu_{K-1})$ define the 2-differential form in the space $\mathbb{R}^{K-1}$ of $B(\mu)$. Then:
\begin{align*}
	\widetilde{B}(\mu_1,\cdots,\mu_{K-1})=&
	\begin{bmatrix}
		I&\mathbf{-1}
	\end{bmatrix}
	\cdot B(\mu_1,\cdots,\mu_{K-1},1-\Sigma \mu_{k})\cdot
	\begin{bmatrix}
		I\\
		\mathbf{-1}^T
	\end{bmatrix}\\
	=&\left[ \frac{Q_{ij}}{\mu_i\mu_j}-\frac{Q_{iK}}{\mu_i(1-\Sigma \mu_k)}-\frac{Q_{jK}}{\mu_j(1-\Sigma \mu_k)}+\frac{Q_{KK}}{(1-\Sigma \mu_k)^2} \right]
\end{align*}
where $I$ is the $(K-1)\times (K-1)$ identity matrix, $\mathbf{1}$ is $K-1$-dimensional column vector. The second line represents the matrix using its $ij^{th}$ element. \cref{thm:ps} states that suppose $\phi(C)$ is (uniformly) posterior separable, $\widetilde{B}$ must be a Hessian matrix of a function $H$ on $\mathbb{R}^{K-1}$. This implies that $ij^{th}$ element of $\widetilde{B}$ is the cross derivative of a $H$ w.r.t. $\mu_i$ and $\mu_j$. Clearly, $\widetilde{B}$ is $C^{\infty}$ smooth on $\Delta(X)^o$. Therefore, suppose $\widetilde{B}=\mathbb{H}H$, then $H$ is also $C^{\infty}$ smooth. Then we can calculate the higher order cross derivative of $H$. Pick any $i\neq j$:
\begin{align*}
	&\frac{\partial^3}{\partial \mu_i^2\partial\mu_j}H(\mu_1,\cdots,\mu_{K-1})=\frac{\partial}{\partial\mu_i}\widetilde{B}(\mu_1,\cdots,\mu_{K-1})_{ij}=\frac{\partial}{\partial\mu_j}\widetilde{B}(\mu_1,\cdots,\mu_{K-1})_{ii}\\
	\iff&\frac{\partial}{\partial\mu_i}\left(\frac{Q_{ij}}{\mu_i\mu_j}-\frac{Q_{iK}}{\mu_i(1-\Sigma \mu_k)}-\frac{Q_{jK}}{\mu_j(1-\Sigma \mu_k)}+\frac{Q_{KK}}{(1-\Sigma \mu_k)^2}\right)= \frac{\partial}{\partial\mu_j}\left(\frac{Q_{ii}}{\mu_i^2}-2\frac{Q_{iK}}{\mu_i(1-\Sigma \mu_k)}+\frac{Q_{KK}}{(1-\Sigma \mu_k)^2}\right)\\
		\iff&-\frac{Q_{ij}}{\mu_i^2\mu_j}-\frac{Q_{jK}}{(1-\Sigma \mu_k)^2\mu_j}+\frac{Q_{iK}}{\mu_i^2(1-\Sigma \mu_k)}=-\frac{Q_{iK}}{\mu_i(1-\Sigma\mu_k)^2}\\
			\iff&Q_{ij}\mu_{K}^2+Q_{jK}\mu_i^2-Q_{iK}(\mu_i+\mu_K)\mu_j=0
\end{align*}
Given the degree of freedom in $\left\{ \mu_k \right\}$, the equality holds only when $Q_{ij}=Q_{iK}=Q_{jK}=0$. This implies $Q$ is a diagonal matrix, which is not permitted by the condition $Q\cdot \mathbf{1}=0$. Therefore there does not exist a non-trivial $Q$ s.t. the corresponding $B(\mu)$ is a Hessian matrix.\par

We sum up the discussion in this subsection by the following proposition:
\begin{prop}\label{prop:prior:independent}
	$\forall$ prior independent $C\in \mathcal{C}$ satisfying \cref{axiom:1,ass:1,ass:2}, $\phi(C)$ is (uniformly) posterior separable \emph{iff} $X$ is binary, and $C(\pi)\ge \E_{\pi}[H(\nu)]-H(\E_{\pi}[\nu])$, where $H$ is defined as in \cref{eqn:H:binary}.
\end{prop}
\cref{prop:prior:independent} completely characterizes direct information costs in the prior independent class that can generate (uniformly) posterior separable indirect cost. The main finding is that generally prior independence can not generate posterior separability --- for the reason that the cost from acquiring incremental evidences to replicate an information structure is generally not path-independent. In other words, the indirect cost generated from prior independent direct cost almost always satisfies strict sub-additivity. Only when the underlying state space is binary, there exists a generic set of direct information costs which generate posterior separable indirect cost. This set of direct costs are bounded below by the standard mutual information measure plus an extra unweighted average KL divergence from signal's marginal distribution to conditional distribution.

\subsection{Indirect information cost in dynamic learning problem}
In this subsection, we study the optimal learning dynamics in a dynamic learning problem where the flow cost of information is captured by an ``indirect cost''. The dynamic learning problem captures the idea that a decision maker purchases information every period (but with a per period budget or incurs increasing marginal cost) and decides when to stop learning and make a decision. The producer of information runs some cost minimization algorithm at background and hence the consumer pays an indirect cost of information effectively.\par

The time horizon $t=0,1,\dots,\infty$ is discrete. The utility associated with action-state pair $(a,x)$ is $u(a,x)$. Where $x\in X$ and $a\in A$ are finite. Equivalently, we use $U(\mu)=\max_{a}\E_{\mu}[u(a,x)]$ to represent the expected decision utility. The decision maker (DM) pays a constant flow cost $m$ for delaying decision by one period. We assume that the DM can choose an arbitrary belief martingale $\langle \mu_t\rangle$ together with a stopping time $\tau$ measurable to $\langle \mu_t\rangle$. The dynamic optimization of the DM is:
\begin{align}
	V(\mu)=\sup_{\langle \mu_t\rangle,\tau}\E\left[ U(\mu_{\tau})-m\tau-\sum_{t=0}^{\tau}f\left( C(\pi_t(\mu_{t+1}|\mu_t)) \right) \right] \label{eqn:dynamic}
\end{align}
where $C(\cdot)\in\phi(\mathcal{C})$ and $f:\mathbb{R}^+\to\overbar{\mathbb{R}}^+$ is increasing and convex.\footnote{$f$ maps to extended real values. So $f$ can also capture a hard cap.}

\begin{prop}\label{prop:dynamic}
	$\forall \mu\in \Delta(X)$:
	\begin{align}
		V(\mu)=\max\left\{ U(\mu),\sup_{\lambda,C(\pi)\ge \lambda}\left( \E_{\pi}[U(\nu)]-\left( \frac{m}{\lambda}+\frac{f(\lambda)}{\lambda} \right)C(\pi) \right) \right\}\label{eqn:dynamic:static}
	\end{align}
	where the superemum is taken over $\lambda$ and $\pi\in\Delta^2(X)$ jointly.\end{prop}

	\cref{prop:dynamic} establishes that solving the optimal value in \cref{eqn:dynamic} is equivalent to a static problem. In the static problem, the DM pays a fixed marginal cost $\left( \frac{m}{\lambda}+\frac{f(\lambda)}{\lambda} \right)$ on each unit of information cost $C(\pi)$. Generically, the constraint $C(\pi)\ge\lambda$ does not bind and optimal $\lambda$ is simply chosen to maximize $\frac{m}{\lambda}+\frac{f(\lambda)}{\lambda}$. \footnote{$\lambda$ captures the optimal per period information cost. The constraint only binds when the optimal $C(\pi)$ is so low that it is optimal to acquire it in ``less than one period''. This case is treated separately to deal with the integer problem.} Once $\lambda$ is fixed, we can solve for optimal $\pi$. Note that due to \cref{axiom:1}, it is without loss of optimality to consider direct signals, namely signal structures directly specifying the conditional distribution of actions. Since $A$ is finite, the optimization problem reduces to a finite dimensional problem. Since $C\in\mathcal{C}$ is continuous, solution to \cref{eqn:dynamic:static} exists. \par
	Once the static problem \cref{eqn:dynamic:static} is solved, we can construct $(\langle \mu_t\rangle,\tau)$ based on its solution to solve \cref{eqn:dynamic}.
	\begin{prop}\label{prop:dynamic:poisson}
		$\forall \mu\in \Delta(X)$, $\forall \pi\in \Delta^2(X)$ s.t. $\E_{\pi}[\nu]=\mu$ and $\forall \lambda^*<C(\pi)$, define $(\langle\mu_t\rangle,\tau)$ as:
		\begin{enumerate}
			\item $
				\begin{dcases}
					\text{Prob}(\tau=t|\tau\ge t)=\frac{\lambda^*}{C(\pi)}\\
					\text{Prob}(\tau>t|\tau\ge t)=1-\frac{\lambda^*}{C(\pi)}
			\end{dcases}$
		\item $\mu_{t+1}
			\begin{dcases}
				=\mu_0&\text{if }\tau>t\\
				\sim \pi&\text{if }\tau=t
			\end{dcases}
			$
		\end{enumerate}
		Then: $\E_{\pi}[U(\nu)]-\left( \frac{m}{\lambda^*}+\frac{f(\lambda^*)}{\lambda^*} \right)C(\pi)=\E\left[ U(\mu_{\tau})-m\tau-\sum_{t=0}^{\tau}f\left( C(\pi_t(\mu_{t+1}|\mu_t)) \right) \right]$.
	\end{prop}
	\cref{prop:dynamic:poisson} shows that \cref{eqn:dynamic} can be solved using a very simple learning strategy. At each period, conditional on not yet stopped, the belief either jumps to posterior according to $\pi$  with probability $\frac{\lambda^*}{C(\pi)}$, or stays at prior $\mu_0$ with remaining probability $1-\frac{\lambda^*}{C(\pi)}$. And learning is stopped once the jump happens. It is easy to verify that such a strategy incurs flow cost $C(\pi_t(\mu_{t+1}|\mu_t))=\frac{\lambda^*}{C(\pi)}\cdot C(\pi)= \lambda^* $ (by \cref{lem:pos}). The expected stopping time is $\frac{C(\pi)}{\lambda^*}$, hence expected delay cost is $\frac{m}{\lambda^*}C(\pi)$ and expected total information cost is $\frac{f(\lambda^*)}{\lambda^*}C(\pi)$. The distribution of $\mu_{\tau}$ is by definition $\pi$ itself. Therefore \cref{prop:dynamic:poisson} is proved. \par
	\cref{prop:dynamic:poisson} directly proves that $V(\mu)$ is weakly larger than the static optimization problem \cref{eqn:dynamic:static}. Intuitively, the static problem captures the expected utility of a specific way to acquire information: ``dilute'' the terminal belief $\pi$ into every period evenly by conducting $\pi$ with small probability like a Poisson process. Now we show that \cref{axiom:2} implies that this is the optimal way to acquire information.\par
	Take any $\left( \langle\mu_t\rangle,\tau \right)$, let $\pi$ be the distribution of $\mu_{\tau}$. Then sub-additivity implies $C(\pi)\le\E\left[\sum_{t=0}^{\tau} C(\pi_{t}(\mu_{t+1}|\mu_t)) \right]$. Then:
	\begin{align*}
		\frac{C(\pi)}{\E[\tau]}\le&\frac{\E\left[\sum_{t=0}^{\tau} C(\pi_{t}(\mu_{t+1}|\mu_t)) \right]}{\E[\tau]}\\
		\implies f\left( \frac{C(\pi)}{\E[\tau]} \right)\le&\frac{\E\left[\sum_{t=0}^{\tau} f(C(\pi_{t}(\mu_{t+1}|\mu_t))) \right]}{\E[\tau]}\\
		\implies \E[\tau] f\left( \frac{C(\pi)}{\E[\tau]} \right)\le &\E\left[\sum_{t=0}^{\tau} f(C(\pi_{t}(\mu_{t+1}|\mu_t))) \right]
	\end{align*}
	The second inequality uses convexity of $f$. In words, sub-additivity implies that acquiring small pieces of information separately is more expensive than acquiring all information at once. However, convex $f$ gives the DM incentive to smooth the information across periods, and the optimal way to smooth is to distribute cost evenly across periods. Therefore, 
	\begin{align*}
		\E_{\pi}\left[U(\nu)-\E[\tau]m-\E[\tau]f\left( \frac{C(\pi)}{\E[\tau]} \right)\right]
	\end{align*}
	is an upper bound for $V(\mu)$. Then \cref{eqn:dynamic:static} follows by relabeling $\frac{C(\pi)}{\E{\tau}}$ with $\lambda$. Moreover, each inequality in the previous analysis can be strengthened to strict inequality if we consider strict convexity, monotonicity and sub-additivity:
	\begin{prop}\label{prop:dynamic:solution}
		$\forall $ solution of \cref{eqn:dynamic}:
		\begin{enumerate}
			\item If $C(\cdot)$ is strictly monotonic, then the support of $\mu_{\tau}$ is unique.
			\item If $f(\cdot)$ is strictly convex, then $C(\pi_t(\mu_{t+1}|\mu_t))$ is constant.
			\item If $C(\cdot)$ is strictly sub-additive, then $(\langle\mu_t\rangle,\tau)$ is unique.
		\end{enumerate}
	\end{prop}\par
	The analysis in this section complements the analysis in \cite{zhong2017dynamic} and \cite{hebert2016rational}, which explores optimal dynamic learning when the flow cost is \emph{additive} with and without discounting, respectively. The two papers illustrates that additive information cost makes all types of information equally optimal without discounting and discounting implies the optimality of Poisson type information. \cref{prop:dynamic:poisson} illustrates that even without discounting, sub-additivity of information cost also implies the optimality of Poisson type information.

	\section{Summary}
	In this paper, we study the indirect cost of information, defined as the lowest total direct cost of sequentially acquire information to replication an information structure. We show that the class of indirect cost functions can be characterized as all direct cost functions satisfying monotonicity and sub-additivity. The class of (uniform) posterior separable indirect cost functions is generated from direct cost functions favoring incremental evidences. We also provide necessary and sufficient condition when prior independent direct cost could generate (uniform) posterior separable indirect cost. In an application, we study the implication when indirect cost is used in a dynamic learning problem and show that Poisson-like learning strategy is optimal.

\newpage
\bibliographystyle{apalike}
\bibliography{report}

\begin{thebibliography}{}

\bibitem[Blackwell et~al., 1951]{blackwell1951comparison}
Blackwell, D. et~al. (1951).
\newblock Comparison of experiments.
\newblock In {\em Proceedings of the second Berkeley symposium on mathematical
  statistics and probability}, volume~1, pages 93--102.

\bibitem[Caplin and Dean, 2013]{caplin2013rational}
Caplin, A. and Dean, M. (2013).
\newblock Behavioral implications of rational inattention with shannon entropy.
\newblock Working Paper 19318, National Bureau of Economic Research.

\bibitem[Caplin and Dean, 2015]{caplin2015revealed}
Caplin, A. and Dean, M. (2015).
\newblock Revealed preference, rational inattention, and costly information
  acquisition.
\newblock {\em The American Economic Review}, 105(7):2183--2203.

\bibitem[Caplin et~al., 2017]{NBERw23652}
Caplin, A., Dean, M., and Leahy, J. (2017).
\newblock Rationally inattentive behavior: Characterizing and generalizing
  shannon entropy.
\newblock Working Paper 23652, National Bureau of Economic Research.

\bibitem[Frankel and Kamenica, 2019]{frankel2018quantifying}
Frankel, A. and Kamenica, E. (2019).
\newblock Quantifying information and uncertainty.
\newblock {\em American Economic Review}, 109(10):3650--80.

\bibitem[H{\'e}bert and Woodford, 2017]{hebert2016rational}
H{\'e}bert, B. and Woodford, M. (2017).
\newblock Rational inattention and sequential information sampling.
\newblock Technical report, National Bureau of Economic Research.

\bibitem[Kullback and Leibler, 1951]{kullback1951information}
Kullback, S. and Leibler, R.~A. (1951).
\newblock On information and sufficiency.
\newblock {\em The annals of mathematical statistics}, 22(1):79--86.

\bibitem[Mensch, 2018]{mensch2018cardinal}
Mensch, J. (2018).
\newblock Cardinal representations of information.
\newblock {\em Available at SSRN 3148954}.

\bibitem[Morris and Strack, 2019]{morris2017wald}
Morris, S. and Strack, P. (2019).
\newblock The wald problem and the equivalence of sequential sampling and
  static information costs.
\newblock Working paper.

\bibitem[Pomatto et~al., 2018]{pomatto2018cost}
Pomatto, L., Strack, P., and Tamuz, O. (2018).
\newblock The cost of information.

\bibitem[Sims, 2003]{sims2003implications}
Sims, C.~A. (2003).
\newblock Implications of rational inattention.
\newblock {\em Journal of monetary Economics}, 50(3):665--690.

\bibitem[{Zhong}, 2019]{zhong2017dynamic}
{Zhong}, W. (2019).
\newblock {Optimal dynamic information acquisition}.

\end{thebibliography}

\newpage
\appendix
\section{Proofs in \cref{sec:characterization}}

\subsection{Key lemmas}

Below we consider a general class of sequential information acquisition process. For a given time horizon $2T$ and prior belief $\mu_0$. The information acquisition process is denoted by complete and separable signal spaces $\left\{ S_t \right\}_{t=1}^{2T}$, where each $S_{2t}$ nests $S_{2t-2}$ and $S_{2t-1}$, and a regular probability measure $f$ on $X\times \prod S_t$. $f$ is defined by:\footnote{``Conditional distributions'' are formally defined using the Raydon-Nikodym derivatives, and ``regularity conditions'' are the usual absolute continuity conditions for their existence.}
\begin{align*}
	f(x,s_1,\cdots,s_{2T})=\mu_0(x)\prod_{t=0}^{T-1} f(s_{2t+1}|s_{2t},x)f(s_{2t+2}|s_{2t+1},s_{2t})
\end{align*}
By definition, $s_{2t+1}$ ``adds'' information about $x$ to $s_{2t}$. Then $s_{2t+2}$ ``discards'' information from $s_{2t+1}$. Since $S_{2t}$'s are assumed to be nested, $s_{2t}$ can potentially record the entire path of $s_0,\cdots,s_{2t-1}$, and hence perfect memory is also admissible. Then the total cost of the process $f$ is:
\begin{align*}
	\sum_{t=0}^{T-1}\E_{s_{2t}} \left[ C( \nu(s_{2t+1}|s_{2t}) ) \right]
\end{align*}

\begin{lem}\label{lem:markov}
	$\forall C\in\mathcal{C}$ satisfying \cref{axiom:0}, $\forall$ $2T$-period signal process $\langle s_t \rangle$ defined as before, there exists $2T$-period $\langle \mu_t \rangle$ satisfying \cref{defi:rep} s.t. the posterior induced by $s_{2t}$ is distributed as $\mu_{2T}$ and:
	\begin{align*}
		\sum_{t=0}^{T-1}\E_{\mu_{2t}}\left[ C(\pi_t(\mu_{2t+1}|\mu_{2t})) \right]\le\sum_{t=0}^{T-1}\E_{s_{2t}} \left[ C( \nu(s_{2t+1}|s_{2t}) ) \right]
	\end{align*}
\end{lem}
\begin{proof}
	Given any process $\langle s_t \rangle$, it induces a joint probability measure $m\left( \mu_0,\mu_1,\cdots,\mu_{2T} \right)$, where each $\mu_{2t}$ is the conditional measure of $x$ on $s_{2t}$ and each $\mu_{2t+1}$ is the conditional measure of $x$ on $\left( s_{2t},s_{2t+1} \right)$. Now we convert this measure to get a process $\langle \widehat{\mu}_t \rangle$ satisfying \cref{defi:rep}. Define joint probability measure $\widehat{m}$ by:
	\begin{align*}
		\begin{dcases}
		m(\mu_{2t},\mu_{2t+1})=m(\mu_{2t})\widehat{m}(\mu_{2t+1}|\mu_{2t})\\
		m(\mu_{2t+1},\mu_{2t+2})=m(\mu_{2t+1})\widehat{m}(\mu_{2t+2}|\mu_{2t+1})
	\end{dcases}
	\end{align*}
	then $\widehat{m}(\mu_0,\mu_1,\cdots,\mu_{2t})=\prod \widehat{m}(\mu_{2t+1}|\mu_{2t})\widehat{m}(\mu_{2t+2}|\mu_{2t+1})$. It is easy to verify by induction that $\widehat{m}(\mu_{2t})=m(\mu_{2t})$. Now we verify that the process $\langle \widehat{\mu}_t \rangle$ defined according to $\widehat{m}$ satisfies the conditions in \cref{defi:rep}. \par
	First, by definition $\langle \widehat{m}_t \rangle$ is Markov. Second, we verify the martingale property. By definition:
	\begin{align*}
		\E_{\widehat{m}}[\mu_{2t+1}|\mu_{2t}]\cdot \widehat{m}(\mu_{2t})=& \int\widehat{m}(\mu_{2t+1},\mu_{2t})\mu_{2t+1}\d \mu_{2t+1}\\
		=&\int m(\nu,\mu_{2t})\nu\d \nu\\
		=&\int \left( \int_{(s_{2t},s_{2t+1})\to(\mu_{2t},\nu)}f(s_{2t},s_{2t+1})\frac{\mu_{2t}(\cdot)f(s_{2t+1}|s_{2t},\cdot)}{f(s_{2t},s_{2t+1})} \d s_{2t},s_{2t+1}\right)\d \nu\\
		=&\int_{s_{2t}\to \mu_{2t}}\int \mu_{2t}(\cdot) f(s_{2t+1}|s_{2t},\cdot)\d s_{2t+1}\d s_{2t}\\
		=&\mu_{2t}\cdot m(\mu_{2t})= \mu_{2t}\cdot \widehat{m}(\mu_{2t})
	\end{align*}
 Notation $s\to \mu$ means $\mu$ is the posterior belief induced by signal $s$. The second equality is by definition of $\widehat{m}$. The third equality is by the Bayes rule that determines $\mu_{2t+1}$. The forth equality is by 1) $s_{2t}$ determines $\mu_{2t}$ 2) Fubini theorem. The last equality is straight forward. 
	\small
	\begin{align*}
		&\E_{\widehat{m}}[\mu_{2t+1}|\mu_{2t+2}]\cdot \widehat{m}(\mu_{2t+2})\\
		=&\int \widehat{m}(\mu_{2t+1},\mu_{2t+2})\mu_{2t+1}\d \mu_{2t+1}\\
		=&\int m(\mu_{2t},\mu_{2t+1})\cdot \frac{m(\mu_{2t+1},\mu_{2t+2})}{m(\mu_{2t+1})}\mu_{2t+1}\d \mu_{2t},\mu_{2t+1}\\
		=&\int m(\nu,\mu_{2t+2}) \nu \d \nu\\
	=&\int\left(\int \left( \int_{(s_{2t},s_{2t+1},s_{2t+2})\to(\mu,\nu,\mu_{2t+2})} f(s_{2t},s_{2t+1},s_{2t+2})\cdot \frac{\mu(\cdot)f(s_{2t+1}|s_{2t},\cdot)}{f(s_{2t},s_{2t+1})} \d s_{2t},s_{2t+1},s_{2t+2} \right)\d \mu \right)\d \nu\\
	=&\int\left(\int \left( \int_{(s_{2t},s_{2t+1},s_{2t+2})\to(\mu,\nu,\mu_{2t+2})} f(s_{2t},s_{2t+1},s_{2t+2})\cdot \mu_{2t+2} \d s_{2t},s_{2t+1},s_{2t+2} \right)\d \mu \right)\d \nu\\
	=&\mu_{2t+2}\cdot \widehat{m}(\mu_{2t+2})
	\end{align*}
	\normalsize
	The second equality is by definition of $\widehat{m}$. The forth equality is by the Bayes rule that determines $\mu_{2t+1}$. The fifth equality is by the Bayes rule that determines $\mu_{2t+2}$. The last equality is straight forward. Moreover, by definition $\widehat{m}$ always has the same marginal distribution as $m$. So the distributions of induced belief at period $2T$ are the same. Therefore, $\widehat{m}$ defines a process $\langle \widehat{\mu}_t \rangle$ satisfying \cref{defi:rep}.\par

	Now we show that the cost of $\langle \widehat{\mu}_t \rangle$ is weakly lower than that of $\langle s_t \rangle$:
	\begin{align*}
		\E_{s_{2t}}\left[ C(\nu(s_{2t+1}|s_{2t})) \right]=& \E_{\mu_{2t}}\left[ \E_{s_{2t}}\left[C( \nu(s_{2t+1}|s_{2t}) ) |\mu_{2t} \right] \right]\\
		\ge&\E_{\mu_{2t}}\left[ C\left( \E_{s_{2t}}[\nu(s_{2t+1}|s_{2t})|\mu_{2t}] \right) \right]\\
		=&\E_{\mu_{2t}}\left[ C(\widehat{m}(\cdot|\mu_{2t})) \right]
	\end{align*}
	The inequality is by the second property of \cref{axiom:0}.
\end{proof}

\subsection{Proof of \cref{thm:1}}
\begin{proof} 
	First, we verify the continuity of $C^*(\pi)$. By Prokhorov theorem, $\Delta^2(X)$ is a compact and separable metric space equipped with the L\'{e}vy-Prokhorov metric (henceforth, \emph{l-p} metric). Wlog, we consider the open set of generic information structures in $\Delta^2(X)$. The analysis applies to any generic subspace. $\forall \pi$, since the set of generic information structures is open, there exists an interior closed ball $B_{\delta_0}(\pi)$. Then, Heine-Cantor theorem implies that $C(\pi)$ is uniformly continuous on $B_{\delta_0}(\pi)$. $\forall \epsilon>0$, $\forall \pi \in \Delta^2(x)$, let $\delta<\delta_0$ be the uniform continuity parameter of $C$ w.r.t. $\epsilon'<\epsilon$.\par
	\emph{Upper semi-continuity}: $\forall \pi$ and $\pi'\in B_{\delta}(\pi)$. Let $\mu_0=\E_{\pi}[\nu]$ and $\mu'_0=\E_{\pi'}[\nu]$. Then $\|\mu_0-\mu'_0\|< 2\delta$.\footnote{$\pi'-\pi$ can be written as $m^+-m^-$ where both are positive measure and bounded by $\delta$ by the definition of $l-p$ metric. Then $\|\mu_0-\mu_0'\|=\|\int \nu (\d m^+- \d m^-)\|\le 2\delta$. } Since $B_{\delta_0}(\pi)$ is interior, $\delta$ can be picked small enough s.t. $\forall \pi'$, there exists $\nu$ s.t. $\mu_0'=\alpha \nu+(1-\alpha)\mu_0$ and $\alpha<\epsilon'$.\footnote{A formal proof should be added.} Now pick $\langle \mu_t \rangle\rep \pi$ and with totally cost lower than $C^*(\pi)+\epsilon$, construct a sequential learning strategy for $\pi'$: 1) acquire some information and get posterior $\mu_0$ and $\nu$, 2) conditional on $\mu_0$, follow $\langle \mu_0 \rangle$. 3) contract terminal belief by mixing $\mu_{2T}$ and $\nu$ with probability $1-\alpha$ and $\alpha$ and get $\mu_{2T+1}$. By construction, $\mu_{2T+1}$ shifts $\pi$ by $(\mu_0'-\mu_0)$. Let $\pi''$ be the distribution of $\mu_{2T+1}$. Then $\pi''$ has the same mean as $\pi'$ and $d_{l-p}(\pi',\pi'')\le d_{l-p}(\pi,\pi')+d_{l-p}(\pi,\pi'')<3\delta$. Denote $\pi'-\pi''=m^+-m^-$ where both are positive measures, bounded by $3\delta$ and satisfy $\E_{m^+}[\nu]=\E_{m^-}[\nu]$. 4) contract $m^-$ and then acquire posterior according to $m^+$. Then we replicated $\pi'$ through 1)-4). Count the total cost: step 1) acquires information structure within $B_{2\delta+\epsilon'}(\delta_{\mu_0})$; step 2) incurs cost weakly less than $C^*(\pi)+\epsilon$; step 3 incurs zero cost and step 4) acquires some information with less than $3\delta$ probability. By \cref{lem:markov}, this process can always be modified to satisfy \cref{defi:rep}, replicate $\pi$ and has weakly lower cost. Therefore, if we choose $\epsilon'$ sufficiently small, the total cost is bounded above by $C^*(\pi)+3\epsilon$, hence $\varlimsup_{\pi'\to\pi}C^*(\pi')\le C^*(\pi)$. \par
	\emph{lower semi-continuity}: $\delta$ can be picked sufficiently small that $\forall \mu_0$, there also exists $\nu'$ and $\alpha'$ s.t. $\mu_0=\alpha\nu'+(1-\alpha)\mu_0' $. Then previous argument also shows that $C^*(\pi)\le \varliminf_{\pi'\to \pi} C^*(\pi')$. \par
	Therefore, since $C^*(\pi)$ is both upper semi-continuous and lower semi-continuous, $C^*(\pi)$ is continuous at any generic $\pi$. Since $0\le C^*(\pi)\le C(\pi)$, $C^*(\pi)$ is bounded.\par
	
	Second, we show that $\forall C^*\in \phi(\mathcal{C})$, $C^*$ satisfies \cref{axiom:0,axiom:1,axiom:2}. \par
	\cref{axiom:0}: Since $0\le C^*(\pi)\le C(\pi)$, it is trivial that $C^*(\delta_{\mu})\equiv 0$. $\forall \mu_0\in \Delta(X)^o$, consider $A=\left\{ \pi\in \Delta^2(X)|\E_{\pi}[\nu]=\mu_0 \right\}$. $A$ is a compact and separable subset of generic information structures. $\forall \epsilon$, there exists a finite $\epsilon$-net of $A$. Now $\forall P\in \Delta(A)$, discretizing $P$ on the $\epsilon$-net gives finite distribution $\widehat{P}$ $\epsilon$-close to $P$ (under $l-p$ metric). Therefore, there exists finite distributions $\widehat{P}$ converging to $P$.  Now given $\widehat{P}$, $\forall \epsilon$, there exists a uniform upper-bound $T$ for all $\pi$ in $\text{supp} (\widehat{P})$ such that $\langle \mu_t \rangle_{t=0}^{2T}$ replicates $\pi$ and the total cost is lower than $C^*(\pi)+\epsilon$. This implies $C^*(\E_{\widehat{P}})\le \E_{\widehat{P}}[C^*(\pi)]+\epsilon$. By continuity of $C^*$ (shown in the first part), $\E_{\widehat{P}}[C^*(\pi)]\to \E_{P}[C^*(\pi)]$ and $C^*(\E_{\widehat{P}}[\pi])\to C^*(\E_{P}[\pi])$. To sum up, $C^*(\E_P[\pi])\le\E_{P}[C^*(\pi)]$.
	\par
	\cref{axiom:1}: $\forall \pi,\pi'\in \Delta^2(X)$ and $\pi\le_{BW}\pi'$, by definition, there exists $\pi''(\nu|\mu)$ s.t. $\pi'(\nu)=\E[\pi(\mu)\pi''(\nu|\mu)]$ and $\E[\pi''(\nu|\mu)]=\mu$. From joint distribution $\pi(\mu)\pi''(\nu|\mu)$, we can obtain marginal distribution $\widehat{\pi}(\mu|\nu)$. Now $\forall$ 2T-period $ \langle \mu_t \rangle$ replicating $\pi'$, define $2T$-period $\langle \widehat{\mu}_t \rangle$ replicating $\pi$: $\widehat{\mu}_t=\mu_t$ when $t<2T$ and $\widehat{\mu}_{2T}|\widehat{\mu}_{2T-1}\sim \E\left[ \widehat{\pi}(\widehat{\mu}_{2T}|\mu_{2T})|\widehat{\mu}_{2T-1} \right]$. It is easy to verify that $\langle \widehat{\mu}_t \rangle$ satisfies the conditions in \cref{defi:rep} and hence $\langle \widehat{\mu}_t \rangle$ replicates $\pi$. Noticing that $\sum C(\pi_t(\mu_{2t+1}|\mu_{2t}))=\sum C(\pi_t(\widehat{\mu}_{2t+1}|\widehat{\mu}_{2t}))$ and therefore $C^*(\pi)\le C^*(\pi') $. \cref{axiom:1} is verified.\par
	\cref{axiom:2}: $\forall \pi(\mu),\pi'(\nu|\mu)$ and $\pi''(\nu)=\E_{\pi}[\pi'(\nu|\mu)]$. $\forall \epsilon>0$. Pick any $\delta>0$ and take the closures of $\delta$-interior points of all $\Delta(X')$'s, denote it by $D^{\delta}$. Then open set $\Delta(X)\setminus D^{\delta}$ is shrinking to an empty set and hence there exists $\delta$ s.t. $\pi(D^{\delta})>1-\epsilon$. Now we construct a sequence of information structures that replicates $\pi''$. \par
	First, let $\mu_0=\E_{\pi}[\nu]$. Let $\mu_{0}'=\E_{\pi}[\nu|\nu\in D^{\delta}]$ and $\mu_{0}''=\E_{\pi}[\nu|\nu\not\in D^{\delta}]$. Then $\mu_0=\pi(D^{\delta})\mu_0'+\pi(\Delta(X)\setminus D^{\delta})\mu_0''$. Define information structure $\pi_0^{\delta}$ with support $\left\{ \mu_0',\mu_0'' \right\}$ and the corresponding probabilities. Define information structures $\pi'^{\delta}_0(\nu)=\pi(\nu|\nu\in \Delta(X)\setminus D^{\delta})$. Now partition $D^{\delta}$ to finite Borel subsets each of diameter $\eta<\delta$, denote the partition by $\left\{ D_i^{\delta,\eta} \right\}$. Define $\widetilde{\pi}^{\delta,\eta}$ with support $\left\{ \nu_i=\E_{\pi}[\nu|\nu\in D_i^{\delta,\eta}] \right\}$ and distribution $\widetilde{\pi}^{\delta,\eta}(\nu_i)=\pi(D_i^{\delta,\eta})$. Define $\widetilde{\pi}_i'{}^{\delta,\eta}=\E_{\pi}\left[ \pi'(\nu|\mu)|\mu\in D^{\delta,\eta}_i \right]$. Now consider the following sequential information structure: First acquire $\pi_0^{\delta}$, then acquire $\pi_0'^{\delta}$ and then $\pi'$ conditional on $\mu_{0}''$, acquire $\widetilde{\pi}^{\delta,\eta}$ conditional on $\mu_0'$. Then following $\widetilde{\pi}^{\delta,\eta}$, acquire $\widetilde{\pi}_i'{}^{\delta,\eta}$ conditional on $\nu_i$. New we verify that the sequential information structure replicates $\pi''$: $\forall$ Borel set $U\subset \Delta(X)$,
	\begin{align*}
		\text{Prob}(U)=&\text{Prob}(U|\mu_0')\pi_0^{\delta}(\mu_0')+\text{Prob}(U|\mu_0'')\pi_0^{\delta}(\mu_0'')\\
		=&\sum_{i}\text{Prob}(U|\nu_i,\mu_0')\widetilde{\pi}^{\delta,\eta}(\nu_i)\pi_0^{\delta}(\mu_0')+\E_{\pi_0'{}^{\delta}}[\pi'(U|\mu)]\pi_0^{\delta}(\mu_0'')\\
		=&\sum_i\widetilde{\pi}_i'{}^{\delta,\eta}(U)\widetilde{\pi}^{\delta,\eta}(\nu_i)\pi_0^{\delta}(\mu_0')+\E_{\pi}\left[ \pi'(U|\mu)|\mu\in \Delta(X)\setminus D^{\delta} \right]\pi(\Delta(X)\setminus D^{\delta})\\
		=&\sum_i\E_{\pi}\left[ \pi'(U|\mu)|\mu\in D_i^{\delta,\eta} \right]\pi(D_i^{\delta,\eta})\pi(D^{\delta})+\E_{\pi}\left[ \pi'(U|\mu)|\mu\in \Delta(X)\setminus D^{\delta} \right]\pi(\Delta(X)\setminus D^{\delta})\\
		=&\E_{\pi}[\pi'(\nu|\mu)]=\pi''(U)
	\end{align*}
	By definition, when $\delta,\eta\to 0$, $\pi_0^{\delta}\xrightarrow{w-*} \delta_{\mu_0}$, $\widetilde{\pi}^{\delta,\eta}\xrightarrow{w-*} \pi$. By continuity of $C$ and $C^*$, $C(\pi_0^{\delta})\to 0$ and $C^*(\widetilde{\pi}^{\delta,\eta})\to C^*(\pi)$. Now we calculate the cost of $\widetilde{\pi}'_i{}^{\delta,\eta}$. $\forall i$, $\forall \mu\in D_i^{\delta,\eta}$, by definition of $\delta$ and $\eta$, $\|\mu-\nu_i\|\le \eta$ and there exists $\mu'\in \Delta(X)$ s.t. $\|\mu-\mu'\|\ge \delta$ and $\nu_i$ is a linear combination of $\mu,\mu'$. Define information structure $\pi'_{i,\mu}(\cdot)=\frac{\|\mu'-\nu_i\|}{\|\mu'-\mu\|} \pi'(\cdot|\mu)+\frac{\|\nu_i-\mu\|}{\|\mu'-\mu\|}\delta_{\mu'}$. Then $d_{l-p}(\pi'(\cdot|\mu),\pi'_{i,\mu}(\cdot))\le \frac{\eta}{\eta+\delta}$. Now consider information structure $\widehat{\pi}_i'{}^{\delta,\eta}=\E_{\pi}[\pi'_{i,\mu}(\nu)|\mu\in D_{i}^{\delta,\eta}]$, then $d_{l-p}(\widetilde{\pi}'_i{}^{\delta,\eta},\widehat{\pi}'_i{}^{\delta,\eta})\le \frac{\eta}{\eta+\delta}$ (because conditional on each $\mu\in D_i^{\delta,\eta}$, the measure on any Borel set differs by at most $\frac{\eta}{\eta+\delta}$). Since $C^*$ is continuous on $D^{\delta}$, and hence uniformly continuous by Heine-Cantor, there exists $\eta$ s.t. $\frac{\eta}{\eta+\delta}$ is the uniform continuity parameter w.r.t. $\epsilon$ for $C^*$. Therefore:
	\begin{align*}
		\E_{\pi}\left[ C^*(\pi'(\nu|\mu)) |\mu\in D_{i}^{\delta,\eta} \right]\ge&\E_{\pi}\left[ C^*(\pi'_{i,\mu}(\nu))|\mu\in D_i^{\delta,\eta} \right]-\epsilon\\
		\ge&C^*\left( \E_{\pi}\left[ \pi'_{i,\mu}(\nu) |\mu\in D_i^{\delta,\eta} \right] \right)-\epsilon\\
		\ge&C^*\left( \widetilde{\pi}'_i{^{\delta,\eta}} \right)-2\epsilon
	\end{align*}
	The two $\epsilon$ each comes from the distance between $\pi',\pi'_{i,\mu}$ and $\widetilde{\pi}'_i{}^{\delta,\eta},\widehat{\pi}'_i{}^{\delta,\eta}$. The second inequality is implied by \cref{axiom:0}, which is verified before.\par

	Now we construct a belief process replicating $\pi''$ and satisfy \cref{defi:rep}. $\forall \widetilde{\pi}'_i{}^{\delta,\eta}$, there exists a $2T_i$ process $\langle \mu_t^i \rangle\rep \widetilde{\pi}'_i{}^{\delta,\eta}$ such that $\sum C(\pi_t(\mu_{2t+1}^i|\mu_{2t}^i))\le C^*(\widetilde{\pi}'_i{}^{\delta,\eta})+\epsilon$. There also exists a $2T_0$ process $\langle \mu_t^0 \rangle\rep \widetilde{\pi}^{\delta,\eta}$ s.t. $\sum C(\pi_t(\mu_{2t+1}^{0}|\mu_{2t}^0))\le C^*(\widetilde{\pi}^{\delta,\eta})+\epsilon$. Let $T=\max\left\{ T_i \right\}+T_0+1$.\par
	First, let $\mu_{1}\sim \pi_0^{\delta}$, $\mu_2=\mu_1$. If $\mu_2=\mu_0''$, $\mu_3\sim \E_{\pi_0'{}^{\delta}}\left[ \pi'(\nu|\mu) \right]$ and all $\mu_t\equiv \mu_3$ for $t>3$. If $\mu_2=\mu_0'$, $\mu_{t+2}\equiv \mu^0_t$, and $\mu_{t+T_0+2}\equiv \mu_t^i$ conditional on $\mu_{T_0+2}=\nu_i$. The total direct cost of this process is:
	\begin{align*}
		&C(\pi_0^{\delta})+\pi_0^{\delta}(\mu_0'')C\left( \E_{\pi'_0{}^{\delta}}[\pi'(\nu|\mu)] \right)\\
		&+\pi_0^{\delta}(\mu_0')\left( \sum_{t=0}^{2T_0}C(\pi_t(\mu^0_{2t+1}|\mu^0_{2t}))+\sum_i\widetilde{\pi}^{\delta,\eta}(\nu_i)\left( \sum_{t=0}^{2T_i}C(\pi_t(\mu_{2t+1}^i|\mu_{2t}^i))\right) \right)\\
		\le&	C(\pi_0^{\delta})+(1-\pi(D^{\delta}))C\left( \E_{\pi'_0{}^{\delta}}[\pi'(\nu|\mu)] \right)\\
			 &+\pi(D^{\delta})\left( C^*(\widetilde{\pi}^{\delta,\eta})+\epsilon+\sum_i\pi(D_i^{\delta,\eta}|D^{\delta})C^*(\widetilde{\pi}'_i{}^{\delta,\eta})+\epsilon \right)\\
		\le&C(\pi_0^{\delta})+(1-\pi(D^{\delta}))C\left( \E_{\pi'_0{}^{\delta}}[\pi'(\nu|\mu)] \right)\\
			 &+\pi(D^{\delta})\left( C^*(\widetilde{\pi}^{\delta,\eta})+ \sum_i\pi(D_i^{\delta,\eta}|D^{\delta})\E_{\pi}[C^*(\pi'(\nu|\mu))|\mu\in D_i^{\delta,\eta}] +4\epsilon\right)\\
		\to&C^*(\pi)+\E_{\pi}[C^*(\pi'(\nu|\mu))]\text{ when }\delta\to0\&\frac{\eta}{\delta}\to 0
	\end{align*}
	By \cref{lem:markov}, the process $\langle \mu_t \rangle$ can always be transformed to one satisfying \cref{defi:rep} with lower total direct cost. This suggests that $C^*(\pi'')\le C^*(\pi)+\E_{\pi}[C^*(\pi'(\nu|\mu))]$.

\end{proof}

\subsection{Proof of \cref{lem:diff}}
\begin{proof}
	We first characterize the local quadratic approximation of $\widehat{C}_n(P,\mu)$ for a given alphabet with size $n$.\par
	First of all, it is easy to see that $C(\delta_{\mu})\equiv 0\ \implies\ \Lambda_n(p,\mu)\equiv 0$, and $C(\cdot)\ge0$ $\implies$ $\Theta_n(p,\mu)\equiv 0$. Now we study the quadratic term. $\forall \mu$, $\forall p\in \Delta(n)^o$. Define an approximately uninformative signal structure $P=p\cdot \mathbf{1}^T+\alpha_1\cdot p_1\cdot e_{x_1}^T+\alpha_2\cdot p_2\cdot e_{x_2}^T$, where $\mathbf{1}$ is the $|X|$-dimensional column vector with all $1$'s, $p_1,p_2$ are $n$-dimensional column vectors that add up to $0$, $e_{x_1},e_{x_2}$ are $|X|$-dimensional column vectors with all $0$'s except $x_i$-th row being $1$.\par
	$\forall \mu_0$, given $\epsilon,\delta$ chosen as in \cref{ass:1}. Fix $p,p_1,p_2$, when $\alpha_1,\alpha_2$ are sufficiently small, $P$ satisfies $\sum p_i\|P_i/p_i-\mathbf{1}\|^2\le \delta$ and $\max\|P_i/p_i-\mathbf{1}\|^2\le \delta$. Therefore, \cref{ass:1} implies $\forall \mu\in B_{\delta}(\mu_0)$:
	\begin{align*}
		\widehat{C}_n(P,\mu)=&\left( \alpha_1\cdot p_1\cdot e_{x_1}^T+\alpha_2\cdot p_2\cdot e_{x_2}^T \right)^T\cdot Q_n(p,\mu)\cdot\left( \alpha_1\cdot p_1\cdot e_{x_1}^T+\alpha_2\cdot p_2\cdot e_{x_2}^T \right)+O\left(\epsilon \sum p_i\|P_i/p_i-\mathbf{1}\|^2 \right)\\
		=&\alpha_1^2p_1^TQ_n^{x_1,x_1}(p,\mu)p_1+\alpha_2^2p_2^TQ_n^{x_2,x_2}(p,\mu)p_2+2x_1x_2p_1^TQ_n^{x_1,x_2}(p,\mu)p_2\\
		 &+O\left( \epsilon\cdot \left( 2\alpha_1^2 \sum \frac{p_{1,i}^2}{p_i} +2\alpha_2^2\sum \frac{p_{2,i}^2}{\pi_i}\right) \right)
	\end{align*}
	This is to say, consider $\widehat{C}_n(P,\mu)$ as a function $G(p_1,p_2,\alpha_1,\alpha_2,n,p,\mu)$, $G$ is twice differentiable in $(\alpha_1,\alpha_2)$ and:
	\begin{align*}
		\frac{\partial^2G(p_1,p_2,\alpha_1,\alpha_2,n,p,\mu)}{\partial\alpha_1\partial\alpha_2}=p_1^TQ_n^{x_1,x_2}(p,\mu)p_2
	\end{align*}
	Define inner product on $\mathbb{R}^{n}$ as $\langle p_1,p_2 \rangle_{n,p,\mu}=\frac{\partial^2G(p_1,p_2,\alpha_1,\alpha_2,n,p,\mu)}{\partial\alpha_1\partial\alpha_2}\big|_{\alpha_1,\alpha_2=0}$. Since $p$ is defined in $\Delta(n)$ with property $\sum p_i\equiv 1$, the space of vectors in $\mathbb{R}^n$ adding up to $0$ defines exactly the tangent space $M_{p}$. Take any invertible Markov matrix $\Omega$ on $\mathbb{R}^{n\times n}$. $\forall U,V\in M_p$:
	\begin{align*}
		\langle \Omega U,\Omega V \rangle_{n,\Omega p,\mu}=&\frac{\partial^2 G(\Omega U,\Omega V,\alpha_1,\alpha_2,n,\Omega p,\mu)}{\partial\alpha_2\partial\alpha_2}\big|_{\alpha_1,\alpha_2=0}\\
		=&\frac{\partial^2}{\partial\alpha_1\partial\alpha_2} \widehat{C}_n\left( \Omega p+\alpha_1\Omega Ue^T_{x_1}+\alpha_2\Omega Ve^T_{x_2},\mu \right)\big|_{\alpha_1,\alpha_2=0}\\
		=&\frac{\partial^2}{\partial\alpha_1\partial\alpha_2} \widehat{C}_n\left( \Omega (p+\alpha_1Ue^T_{x_1}+\alpha_2Ve^T_{x_2}),\mu \right)\big|_{\alpha_1,\alpha_2=0}\\
		=&\frac{\partial^2}{\partial\alpha_1\partial\alpha_2} \widehat{C}_n\left( p+\alpha_1Ue^T_{x_1}+\alpha_2Ve^T_{x_2},\mu \right)\big|_{\alpha_1,\alpha_2=0}\\
		=&\langle U,V \rangle_{n,p,\mu}
	\end{align*}
	The forth equality is from the fact that multiplying the signal structure by invertible Markov matrix $\Omega$ does not change the distribution of induced posterior beliefs and $\widehat{C}_n(P,\mu)$ is defined by $C(\pi)$, which only depends on the distribution of posterior beliefs. Therefore, inner product $\langle \cdot,\cdot \rangle_{n,p,\mu}$ defines a continuous Riemannian metric on $\Delta(n)^o$, and is isometric in congruent embedding by Markov mapping. By C\'{e}ncov theorem, $\langle \cdot,\cdot \rangle_{n,p,\mu}$ is proportional to a Fisher information metric, i.e. exists $q$ s.t.
	\begin{align*}
		\langle p_1,p_2 \rangle_{n,p,\mu}=q\cdot\sum_i \frac{p_{1,i}p_{2,i}}{p_i}
	\end{align*}
	Since the inner product is defined fixing $\mu,n,x_1,x_2$, the constant $q$ depends on all these parameters. Let $q_n^{x_1,x_2}(\mu)$ denote the map from the parameters to $q$. We have representation:
	\begin{align*}
		&p_1^TQ_n^{x_1,x_2}(p,\mu)p_2=\sum_i \frac{p_{1,i}p_{2,i}}{p_i}q^{x_1,x_2}_n(\mu)\\
		\implies&P_1^TQ_n(p,\mu)P_2=\sum_i \frac{1}{p_i}P_{1,i}\cdot Q_n(\mu)\cdot P_{2,i}^T\\
		\implies&\widehat{C}_n(P,\mu)=\sum_i \frac{1}{p_i} \left( P_i^T-p_i\cdot \mathbf{1}^T \right)Q_n(\mu)\left( P_i-p_i\cdot \mathbf{1} \right)+O(\epsilon\sum p_i\|P_i/p_i-\mathbf{1}\|^2)
	\end{align*}
	where $Q_n(\mu)$ is defined as the matrix $\left\{ q_{n}^{x_1,x_2}(\mu) \right\}$ in $\mathbb{R}^{|X|^2}$, $P$ satisfies \cref{ass:1}. $P_i$ is the column vector representing the conditional distribution of signal $i$ on each state. Since $Q_n(p,\mu)$ is positive semi-definite and continuous in $\mu$, $Q_n(\mu)$ is also positive semi-definite and continuous in $\mu$. \par
	Now we show that $C_n(\mu)$ does not depend on $n$. $\forall P\in \mathbb{R}^{2|X|}$ and $p\in \Delta(2)$, $P$'s rows add up to $0$, consider the cost of $P^1=p\cdot \mathbf{1}^T+\eta P$. $\widehat{C}_2(P^1,\mu)\approx \eta^2 \sum \left( P_i/\sqrt{p_i} \right)^TQ_2(\mu)(P_i/\sqrt{p_i}) $. Now for any $n>2$, consider the signal structure where $P^n_1=P^1_1$ and $P^n_i=\frac{1}{n-1}P^1_2$. The cost is 
	\begin{align*}
		\widehat{C}_n(P^n,\mu)\approx& \eta^2 \left( (P_1/\sqrt{p_1})^TQ_n(\mu)(P_1/\sqrt{p_1})+(n-1)\left( \frac{P_2/(n-1)}{\sqrt{p_2/(n-1)}} \right)^TQ_n(\mu)\left( \frac{P_2/(n-1)}{\sqrt{p_2/(n-1)}} \right) \right)\\
		=&\eta^2\sum (P_i/\sqrt{p_i})^TQ_n(\mu)(P_i/\sqrt{p_i})
	\end{align*}
	$P^n$ induces exactly the same distribution of posterior beliefs as $P^1$, therefore $\widehat{C}_2(P^1,\mu)\equiv\widehat{C}_n(P^n,\mu)$. This implies $Q_n(\mu)\equiv Q_2(\mu)$. For simplicity, we denote them by $Q(\mu)$.\par
	To summarize, $\forall \mu_0\in \Delta(X)^o$, given $\epsilon,\delta$ chosen as in \cref{ass:1}, $\forall p,n,\mu\in B_{\delta}(\mu_0)$, $\forall P$ satisfy \cref{ass:1}:
	\begin{align*}
		\left| \widehat{C}_n(P,\mu)-\sum_i \frac{1}{p_i} \left( P_i^T-\mathbf{1}^T\cdot p_i \right)\cdot Q(\mu) \left( P_i-p_i\cdot \mathbf{1} \right) \right|\le \epsilon \sum p_i\|P_i/p_i-\mathbf{1}\|^2
	\end{align*}
	where $Q(\mu)$ is continuous.\footnote{Notice that if we define $P=p'\cdot \mathbf{1}^T$, then $\widehat{C}_n(P,\mu)\equiv 0$ because this information structure is uninformative. Therefore, we can derive a necessary condition on $Q(\mu)$: $Q(\mu)\cdot\mathbf{1}\equiv 0$.\label{fn:C1=0}} Now consider the corresponding distribution of posterior $\pi$ induced by $P$. Pick a $\delta'>0$. $\forall \pi$ with finite support $\left\{ \mu_i \right\}$ and satisfying $\text{supp}(\pi)\subset B_{\delta'(\mu_0)}$. Let $\mu=\E_{\pi}[\mu_i]$. Construct information structure defined as $P_{i,x}=\pi(\mu_i)\mu_i(x)/\mu(x)$. Let $p_i=\pi(\mu_i)$. Then the likelihood ratio $\frac{P_i}{p_i}=\frac{\mu_i}{\mu}$. $\delta'$ can be picked sufficiently small that for all such $\pi$ $\max\|\mu_i/\mu-\mathbf{1}\|^2\le \delta$. Therefore, the previous result implies:
	\begin{align*}
		&\left| C(\pi)-\sum_i \frac{1}{p_i} \left( P_i^T-\mathbf{1}^T\cdot p_i \right)\cdot Q(\mu) \left( P_i-p_i\cdot \mathbf{1} \right) \right|\le \epsilon \sum p_i\|P_i/p_i-\mathbf{1}\|^2\\
		\iff&\left|C(\pi)-\sum_i p_i \left( \frac{\mu_i}{\mu}-1 \right)^T Q(\mu)\left( \frac{\mu_i}{\mu}-1 \right)  \right|\le \epsilon \sum p_i\|\mu_i/\mu-1\|^2\\
			\implies &\left| C(\pi)-\sum_i p_i (\mu_i-\mu)^T\text{diag}(\mu)^{-1}Q(\mu)\text{diag}(\mu)(\mu_i-\mu) \right|\le \epsilon\sum p_i\|\mu_i-\mu\|^2\cdot\max\left\{ \mu(x)^{-2} \right\}
	\end{align*}
	Define $B(\mu)=\text{diag}(\mu)^{-1}Q(\mu)\text{diag}(\mu)$, $\epsilon'=\epsilon\cdot\max\left\{ \mu(x)^{-2} \right\}$ then:
	\begin{align*}
		\left|C(\pi)-\int (\nu-\mu)^TB(\mu)(\nu-\mu)\d\pi(\nu) \right|\le \epsilon'\int \|\nu-\mu\|^2\d \pi(\nu)
	\end{align*}
	Notice that $\forall \mu_0\in \Delta(X)^o$, we can always pick $\delta'$ small enough s.t. $\max\left\{ \mu(x)^{-2} \right\}$ is bounded. Therefore, $\epsilon'$ can always be chosen arbitrarily small.\par
	The last step is to extend the expression to general $\pi\in \Delta(B_{\delta'}(\mu_0))$. This is easy to prove by observing that both $(\nu-\mu)^TB(\mu)(\nu-\mu)$ and $\|\nu-\mu\|^2$ are continuous functions on $\Delta^2(X)$. Therefore, both sides of the inequality converges  when $\pi$ weakly converges. So any general $\pi$ can be approximated by $\pi$ with finite support.\par
	By \cref{fn:C1=0}, $Q(\mu)\cdot \mathbf{1}\equiv 0$ $\iff$ $B(\mu)\cdot \mu\equiv 0$. Which means not all positive semi-definite matrix $B(\mu)$ can be generated by some $Q(\mu)$. However $\forall$ PSD matrix $B(\mu)$, we can always normalize it to $\widetilde{B}(\mu)=B(\mu)-(\mu^TB(\mu)\mu)\mathbf{1}\cdot\mathbf{1}^T$, without affecting the value of any $(\nu-\mu)^T\widetilde{B}(\mu)(\nu-\mu)$ and $\widetilde{B}(\mu)\cdot \mu=0$.\footnote{Proof is simple by observing that $\mu^T\widetilde{B}(\mu)\mu=0$. And any other vector $\mu'$ satisfies $\mu'\widetilde{B}(\mu)\mu'=(\mu'-\mu)^TB(\mu)(\mu'-\mu)\ge0$. So $\widetilde{B}(\mu)$ is also PSD and $\mu$ is its eigenvector with eigenvalue $0$.} 
\end{proof}

\subsection{Proof of \cref{thm:ps}}

\begin{lem}\label{lem:divergence}
	Suppose $C(\pi)=\int D(\nu||\mu)\d \pi(\nu)\subset \phi(\mathcal{C})$. $D(\nu||\mu)$ is continuously differentiable in $\mu$ and twice differentiable in $\nu$, then $C(\pi)$ is uniformly posterior separable.
\end{lem}
\begin{proof}
	By the assumption of the lemma, $C(\pi)$ is an indirect information cost function. $\forall \pi$ with binary support $\left\{ \mu',\mu'' \right\}$, $\forall \pi'(\nu|\mu)$ s.t. $\E_{\pi'(\cdot|\mu')}[\nu]=\mu'$ and $\pi'(\cdot|\mu'')=\delta_{\mu''}(\cdot)$. Define $\pi''=\E_{\pi}[\pi'(\nu|\mu)]$. Then 
	\begin{align*}
		&C(\pi'')=\pi(\mu'')D(\mu''||\mu)+\pi(\mu')\int D(\nu||\mu)\d \pi'(\nu|\mu')\\
		\le&C(\pi)+\E_{\pi}[C(\pi')]=\pi(\mu')D(\mu'||\mu)+\pi(\mu'')D(\mu''||\mu)+\pi(\mu')\int D(\nu||\mu')\d \pi'(\nu|\mu')\\
		\implies& 0\le D(\mu'||\mu)+\int \left( D(\nu||\mu')-D(\nu||\mu) \right)\d \pi'(\nu|\mu')
	\end{align*}
	The inequality is by \cref{axiom:2}. Notice that $D(\nu||\mu)\ge 0$ implies $\text{J}_{\mu}D(\nu||\mu)|_{\mu=\nu}=0$. Therefore, $\forall \mu,\mu'$ close to each other:
	\begin{align*}
		0\le o(|\mu-\mu'|)+\int \text{J}_{\mu}D(\nu||\mu)|_{\mu=\mu'}\cdot(\mu-\mu')+o(|\mu-\mu'|)\d\pi'(\nu|\mu')
	\end{align*}
	Since the inequality holds for all $\mu,\mu'$ close to each other, this implies $\int J_{\mu}D(\nu||\mu)|_{\mu=\mu'}\d \pi'(\nu|\mu')=0$. Notice that a continuous function on $\Delta(X)$ averages to $0$ for any distribution with prior $\mu'$, this implies the continuous function's graph's convex hull is itself. Therefore the function must be both convex and concave, hence linear in $\nu$. Therefore, there exists $|X|^2$ dimensional matrix valued function $A(\mu')$ s.t.
	\begin{align*}
		&\text{J}_{\mu}D(\nu||\mu)|_{\mu=\mu'}=A(\mu') \cdot (\nu-\mu')\\
		\implies&D(\nu||\mu)=\int_{\nu}^{\mu}A(\mu')(\nu-\mu')\d \mu'
	\end{align*}
	Now consider the Hessian matrix of $D(\nu||\mu)$ w.r.t $\nu$: 
	\begin{align*}
		\text{H}_{\nu}D(\nu||\mu)=&\text{H}_{\nu}\left(\int_{\nu}^{\mu}A(\mu')(\nu-\mu')\d \mu'\right)\\
		=&-A(\nu)
	\end{align*}
	Notice that $A(\nu)$ does not depend on $\mu$. Therefore, $D(\nu||\mu)$ can be represented as $H(\nu)+L(\mu,\nu)$. Where $H$ has Hessian $A$ and $L$ is linear in $\nu$. Now $C(\pi)=\int D(\nu||\mu)\d\pi(\nu)=\int H(\nu)+L(\mu,\nu)\d \pi(\nu)=\E_{\pi}[H(\nu)]+L(\mu,\mu)$. \cref{axiom:0} implies $C(\delta_{\mu})=H(\mu)+L(\mu,\nu)=0\implies L(\mu,\mu)\equiv -H(\mu)$. Therefore, $C(\mu)=\E_{\pi}[H(\mu)]-H(\E_{\pi}[\mu])$.
\end{proof}

Now we prove \cref{thm:ps}

\begin{proof}
	\textbf{Necessity}: We first prove that uniformly posterior separable $\phi(C)$ implies $C$ favoring incremental evidences. Let $C^*=\phi(C)$ be defined as $C^*(\pi)=\E_{\pi}[H(\nu)]-H(\E_{\pi}[\nu])$. Before the main proof, we show two important auxiliary results.\par
	$\forall \pi\in \mathcal{I}(X)$, by definition of \cref{eqn:1}, $\forall \epsilon>0$ there exists $\langle \mu_t \rangle\rep \pi$, s.t. $C^*(\pi)\ge \E\left[ \sum C(\pi_{2t}(\mu_{2t+1}|\mu_{2t})) \right]-\epsilon$. Now consider the sequence $\langle \mu_t \rangle$. The total cost of $\langle \mu_t \rangle$ evaluated using cost function $C^*$ is:
	\begin{align*}
		\E\left[ \sum C^*(\pi_{2t}(\mu_{2t+1}|\mu_{2t})) \right]=&\E\left[ \sum \E_{\pi_{2t}}\left[ H(\mu_{2t+1})-H(\mu_{2t})|\mu_{2t} \right] \right]\\
		=&\E\left[ \sum \E_{\pi_{t}}\left[ H(\mu_{t+1})-H(\mu_{t})|\mu_{t} \right] \right]-\E\left[ \sum \E_{\pi_{2t+1}}\left[ H(\mu_{2t+2})-H(\mu_{2t+1})|\mu_{2t+1} \right] \right]\\
		=&\E\left[ \sum\left( H(\mu_{t+1})-H(\mu_t) \right) \right]-\E\left[ \sum \E_{\pi_{2t+1}}\left[ H(\mu_{2t+2})-H(\mu_{2t+1})|\mu_{2t+1} \right] \right]\\
		=&\E\left[ H(\mu_{2T})-H(\mu_0) \right]]-\E\left[ \sum \E_{\pi_{2t+1}}\left[ H(\mu_{2t+2})-H(\mu_{2t+1})|\mu_{2t+1} \right] \right]\\
		\ge& \E\left[ \sum C(\pi_{2t}(\mu_{2t+1}|\mu_{2t})) \right]-\epsilon-\E\left[ \sum \E_{\pi_{2t+1}}\left[ H(\mu_{2t+2})-H(\mu_{2t+1})|\mu_{2t+1} \right] \right]
	\end{align*}
	By definition $C^*\le C$, therefore this implies:
	\begin{align*}
		\E\left[ \sum(H(\mu_{2t+1})-H(\mu_{2t+2})) \right]\le \epsilon
	\end{align*}
	By \cref{ass:2}, $ C^*(\text{Dist}(\mu_{2t+1}|\mu_{2t+2})) \ge C(\text{Dist}(\mu_{2t+1}|\mu_{2t+2}))\ge m\E\left[ \|\mu_{2t+1}-\mu_{2t+2}\|^2 \right] $, where Dist$(\mu_{2t+1}|\mu_{2t+1})$ denotes the conditional distribution of $\mu_{2t+1}$ on $\mu_{2t+2}$. Combining the two inequalities:
	\begin{align}
		\E\left[ \sum \|\mu_{2t+2}-\mu_{2t+1}\|^2 \right]\le \frac{\epsilon}{m}\label{eqn:thm2:1}
	\end{align}
	\cref{eqn:thm2:1} provides an upper bound for the amount of ``discarded'' information. Its total variance must be bounded above by $\frac{\epsilon}{m}$, which can be arbitrarily small when we choose $\epsilon$ small. The key intuition here is that by \cref{ass:2}, all information is costly, including those being discarded in the end. $C^*$ being posterior separable means those costs on discarded information are avoidable---if we replace each information structure in the sequence with an replicating process. Therefore, total amount of discarded cost is bounded above by the approximation error of the replicating processes. \par
	In the last result, we bound the total size of discarded information. Now, we show that for any replicating process that approximates $C^*$, the probability of a path leaving a small neighbourhood around $\mu_0$ is bounded. $\forall \pi\in \mathcal{I}(X)$, let $\mu_0=\E_{\pi}[\nu]$. Suppose the diameter of supp$(\pi)$ is less than $\delta_0$. Pick an arbitrary $\epsilon>0$ and consider $\langle \mu_{t} \rangle\rep \pi$ and $C^*(\pi)\ge \E\left[ \sum C(\pi_{2t}(\mu_{2t+1}|\mu_2t)) \right]-\epsilon$. The previous analysis implies \cref{eqn:thm2:1}: $\E\left[ \sum \|\mu_{2t+2}-\mu_{2t+1}\|^2 \right]\le \frac{\epsilon}{m}$. Now take any path of $\langle \mu_t \rangle$, denoted by $\mu_t[\omega]$, such that $\mu_{t_0}[\omega]$ first leaves $B_{\delta_1}(\mu_0)$ at period $t$ (here we choose $\delta_1>\delta_0$). In other words, $\forall t<t_0,\  \mu_t[\omega]\in B_{\delta_1}(\mu_0)$ and $\mu_{t_0}[\omega]\not\in B_{\delta_1}(\mu_0)$. Collect all paths $\mu_t[\omega']$ s.t. $\mu_{t}[\omega']=\mu_t[\omega]$ when $t\le t_0$. Let $\Omega_0$ denote the set of events corresponding to theses paths.\par
	Now we construct a process $\langle \widehat{\mu}_t \rangle$ in $\mathbb{R}^{|X|}$, satisfying $\sum_x \widehat{\mu}_t(x)\equiv 1$. The process is defined on event space $\Omega_0$ (with corresponding sigma algebra $\mathcal{F}_0$ and probability measure $P_0$ restricted to $\Omega_0$). For notational simplicity, I label $\langle \widehat{\mu}_t \rangle$ using $t$ from $t_0$ to $2T$. If $t_0$ is even, let $T_0=T+\frac{t_0}{2}$:
	\begin{align*}
		\begin{dcases}
			\widehat{\mu}_{t_0+s+1}[\omega]-\widehat{\mu}_{t_0+s}[\omega]=\mu_{t_0+2s+1}[\omega]-\mu_{t_0+2s}[\omega]\\
			\widehat{\mu}_{T_0+s+1}[\omega]-\widehat{\mu}_{T_0+s}[\omega]=\mu_{t_0+2s+2}[\omega]-\mu_{t_0+2s+1}[\omega]
		\end{dcases}
	\end{align*}
	where $s$ is from $0$ to $T-\frac{t_0}{2}-1$. If $t_0$ is odd, let $T_0=T+\frac{t_0-1}{2}$:
	\begin{align*}
		\begin{dcases}
			\widehat{\mu}_{t_0+s+1}[\omega]-\widehat{\mu}_{t_0+s}[\omega]=\mu_{t_0+2s}[\omega]-\mu_{t_0+2s-1}[\omega]\\
			\widehat{\mu}_{T_0+s+1}[\omega]-\widehat{\mu}_{T_0+s}[\omega]=\mu_{t_0+2s+1}[\omega]-\mu_{t_0+2s}[\omega]
		\end{dcases}
	\end{align*}
	where $s$ is from $0$ to $T-\frac{t_0-1}{2}$. $\langle \widehat{\mu}_t \rangle$ essentially reorders the belief changes of $\langle \mu_{t} \rangle$ by grouping all even periods (acquisition periods) together and then all odd periods (disposal periods) together. Now we verify that $\langle \widehat{\mu}_t \rangle$ is a martingale up to period $T_0$. We only show the case with even $t_0$ and the case with odd $t_0$ follows:
	\begin{align*}
		&\E[\widehat{\mu}_{t_0+s+1}-\widehat{\mu}_{t_0+s}|\widehat{\mu}_{t_0},\cdots,\widehat{\mu}_{t_0+s}]\\
		=&\int \widehat{\mu}_{t_0+s+1}[\omega]-\widehat{\mu}_{t_0+s}[\omega] \d P_0(\omega|(\widehat{\mu}_{t_0},\ldots,\widehat{\mu}_{t_0+s})[\omega]=\widehat{\mu}_{t_0},\cdots,\widehat{\mu}_{t_0+s})\\
		=&\int  (\mu_{t_0+2s+1}[\omega]-\mu_{t_0+2s}[\omega]) \d P_0(\omega|(\widehat{\mu}_{t_0},\ldots,\widehat{\mu}_{t_0+s})[\omega]=\widehat{\mu}_{t_0},\cdots,\widehat{\mu}_{t_0+s})\\
	=&\int  (\mu_{t_0+2s+1}[\omega]-\mu_{t_0+2s}[\omega]) \d P_0(\omega|\mu_{t_0+2s'+1}[\omega]-\mu_{t_0+2s'}[\omega]=\widehat{\mu}_{t_0+s'+1}-\widehat{\mu}_{t_0+s'})\\
	=&\int \left( \int \mu_{t_0+2s+1}[\omega]-\mu_{t_0+2s}[\omega]\d P_0(\omega)\big|(\mu_{t_0},\cdots,\mu_{t_0+2s})[\omega]= \mu_{t_0},\cdots,\mu_{t_0+2s} \right) \\
	 &\d P_0((\mu_{t_0},\cdots,\mu_{t_0+2s})[\omega]= \mu_{t_0},\cdots,\mu_{t_0+2s}\big|\mu_{t_0+2s'+1}[\omega]-\mu_{t_0+2s'}[\omega]=\widehat{\mu}_{t_0+s'+1}-\widehat{\mu}_{t_0+s'} )\\
	=&0
	\end{align*}
	The last equality is by the Markov property of $\langle \mu_t \rangle$ and martingale property of $\langle \mu_{t} \rangle$ at even $t$'s. Therefore, $\langle \widehat{\mu}_t \rangle_{t=t_0,\dots,T_0}$ is a martingale process.\par
	
	Since the previous analysis is done in the event space $\Omega_0$ where $\mu_t$ first crosses $B_{\delta_1}(\mu_0)$ at $\mu_{t_0}$. $\Omega$ can actually been partitioned into $\Omega_0^{\alpha}$'s plus $\Omega_1$, where each $\alpha$ indexes a path first crossing $B_{\delta_1}(\mu_0)$,\footnote{$\Omega_0^{\alpha}$'s are clearly disjoint since a path can only first cross $B_{\delta_1}(\mu_0)$ once.} and $\Omega_1$ contains events when the path never crosses $B_{\delta_1}(\mu_0)$. Let $t_0(\alpha)$ denote the first crossing time of each $\alpha$.
	
	Now we calculate the total amount of information discarded in event space $\Omega_0$. Again we only show the case for $t_0$ even.
	\begin{align*}
		\E\left[\sum_{t=0}^T \|\mu_{2t+2}-\mu_{2t+1}\|^2 \right]=&\E\left[ \left\| \sum_{t=0}^{T}(\mu_{2t+2}-\mu_{2t+1}) \right\|^2 \right] \\
		\ge&\E\left[ \left\| \sum_{t=0}^{T}(\mu_{2t+2}-\mu_{2t+1}) \right\|^2 \Bigg| \cup \Omega_0^{\alpha} \right]\cdot P(\omega\in \cup\Omega_0^{\alpha})\\
		=&\E\left[ \left\|\widehat{\mu}_{2T}-\widehat{\mu}_{T_0}\right\|^2  \Bigg| \cup \Omega_0^{\alpha} \right]\cdot P(\omega\in \cup\Omega_0^{\alpha})\\
		\ge&\E\left[\E\left[ \frac{1}{3}\left( \left\|\widehat{\mu}_{T_0}-\widehat{\mu}_{t_0}\right\|^2+\left\|\widehat{\mu}_{t_0}-\mu_{0}\right\|^2 +\left\|\widehat{\mu}_{2T}-\mu_{0}\right\|^2\right) \Bigg| \Omega_0^{\alpha} \right]\Bigg|\cup\Omega_0^{\alpha}\right]\cdot P(\omega\in \cup\Omega_0^{\alpha})\\
		\ge&\left( \frac{1}{3}\E\left[\E\left[ \left\|\widehat{\mu}_{T_0}-\widehat{\mu}_{t_0}\right\|^2\Big|\Omega_0^{\alpha} \right]\Big| \cup \Omega^{\alpha}_0\right] +\frac{1}{3}\delta_1^2+\frac{1}{3}\delta_0^2 \right)\cdot P\left(\omega\in \cup\Omega_0^{\alpha}\right)\\
		=&\left( \frac{1}{3}\E\left[\E\left[ \sum_{2t\ge t_0(\alpha)}\left\|\mu_{2t+1}-\mu_{2t}\right\|^2\Big|\Omega_0^{\alpha} \right]\Big| \cup \Omega^{\alpha}_0\right] +\frac{1}{3}\delta_1^2+\frac{1}{3}\delta_0^2 \right)\cdot P\left(\omega\in \cup\Omega_0^{\alpha}\right)
	\end{align*}
	The first equality is by $\E\left[ \mu_{2t+1}|\mu_{2t+2} \right]=\mu_{2t+2}$ from \cref{defi:rep}. The first inequality is from the non-negativity of norm. The second equality is by definition of $\langle \widehat{\mu}_t \rangle$. The second inequality is from Cauchy-Schwarz inequality. The last inequality is from $\mu_{t_0}\not\in B_{\delta_1}(\mu_0)$ and definition of $\delta_0$. The last equality is by definition of $\langle \widehat{\mu}_t \rangle$.	Combining the result with \cref{eqn:thm2:1}, we obtain:
\begin{align}
	\left(\E\left[\left\| \sum_{2t\ge t_0(\alpha)}(\mu_{2t+1}-\mu_{2t}) \right\|^2 \Bigg| \cup \Omega_0^{\alpha}\right] +\delta_1^2+\delta_0^2 \right)\cdot P(\omega\in\cup\Omega_0^{\alpha})\le \frac{3\epsilon}{m}
	\label{eqn:thm2:2}
\end{align}\par

Now we state the main proof for \cref{thm:ps}. $\forall \mu_0\in \Delta(X)^o,\ \forall \eta$, let $\delta$ be the parameter pinned down in \cref{ass:diff}. Pick $\delta_0<\delta_1<\delta$. Consider arbitrary $\pi\in \mathcal{I}(X)$ s.t. supp$(\pi)\subset B_{\delta_0}(\mu_0)$ and $\E_{\pi}[\nu]=\mu_0$. Since $B(\mu)$ is continuous, let $\xi$ be $\sup \|B(\mu)-B(\mu_0)\|$ when $\mu\in B_{\delta_1}(\mu_0)$. Now we show that $C^*(\pi)$ is differentiable at $\mu_0$ and:
\begin{align}
\left|C^*(\pi)-\int (\nu-\mu_0)^TB(\mu_0)(\nu-\mu_0)\d \pi(\nu)\right|\le \eta\int \|\nu-\mu_0\|^2 \d \pi(\nu)
	\label{eqn:thm2:3}
\end{align}
where $B(\mu)$ locally characterize $C(\pi)$. Consider any $\langle \mu_t \rangle\rep \pi$ and $\E[\sum C(\mu_{2t+1}|\mu_{2t})]\le C^*(\pi)+\epsilon$ ($\epsilon$ is for now a free parameter that we will pin down later in the proof). The total cost of $\langle \mu_t \rangle$ can be written as:
\begin{align}
	\E[\sum_t C(\mu_{2t+1}|\mu_{2t})]=&\E\left[ \sum_t C(\mu_{2t+1}|\mu_{2t})\mathbf{1}_{\forall t'\le 2t,  \mu_{t'}[\omega]\in B_{\delta_1}(\mu_0)} \right]   \tag{I}\label{eqn:thm2:I}\\
	+&\E\left[ \sum_t C(\mu_{2t+1}|\mu_{2t})\mathbf{1}_{\forall t'< 2t,\mu_{t'}[\omega]\in B_{\delta_1}(\mu_0)\& \mu_{2t}\not\in B_{\delta_1}(\mu_0)} \right]   \tag{II}\label{eqn:thm2:II}\\
	+&\E\left[ \sum_t C(\mu_{2t+1}|\mu_{2t})\mathbf{1}_{\exists t'< 2t,\mu_{t'}[\omega]\not\in B_{\delta_1}(\mu_0)}\right]   \tag{III}\label{eqn:thm2:III}
\end{align}
In words, we partition paths of $\langle \mu_t\rangle$ into three groups. \eqref{eqn:thm2:I} includes paths that never leaves $B_{\delta_1}(\mu_0)$ until $2t$ for each $t$. \eqref{eqn:thm2:II} includes paths that first leaves $B_{\delta_2}(\mu_0)$ at $2t$ for each $t$. \eqref{eqn:thm2:III} includes paths that have left $B_{\delta_1}(\mu_0)$ before $2t$ for each $t$. We bound the total cost of the three groups separately.\par
We first study \cref{eqn:thm2:I}:
\begin{align}
	(\ref{eqn:thm2:I})=&\E\left[ \sum_t C(\mu_{2t+1}|\mu_{2t})\mathbf{1}_{\forall t'\le 2t,  \mu_{t'}[\omega]\in B_{\delta_1}(\mu_0)\& \text{supp}(\mu_{2t+1}|\mu_{2t})\in B_{\delta_1}(\mu_0)} \right]\tag{I-a}\label{eqn:thm2:Ia}\\
	+&\E\left[ \sum_t C(\mu_{2t+1}|\mu_{2t})\mathbf{1}_{\forall t'\le 2t,  \mu_{t'}[\omega]\in B_{\delta_1}(\mu_0)\& \text{supp}(\mu_{2t+1}|\mu_{2t})\not\in B_{\delta_1}(\mu_0)} \right]\tag{I-b}\label{eqn:thm2:Ib}
\end{align}
We further partition group \eqref{eqn:thm2:I} into two sub-groups. In group \eqref{eqn:thm2:Ia}, the paths never leaves $B_{\delta_1}(\mu_0)$ until $2t+1$ and in group \eqref{eqn:thm2:Ib} the paths first leaves $B_{\delta_1}(\mu_0)$ in period $2t+1$.
\begin{align*}
	(\ref{eqn:thm2:Ia})=&\E\Bigg[\sum_t \Bigg(\int (\mu_{2t+1}-\mu_{2t})^TB(\mu_0)(\mu_{2t+1}-\mu_{2t})+(\mu_{2t+1}-\mu_{2t})^T(B(\mu_{2t})-B(\mu_0))(\mu_{2t+1}-\mu_{2t})\\
											&+C(\mu_{2t+1}|\mu_{2t})-(\mu_{2t+1}-\mu_{2t})^TB(\mu_{2t})(\mu_{2t+1}-\mu_{2t})\d\pi_{2t}(\mu_{2t+1}|\mu_{2t})\Bigg)\\
											&\times\mathbf{1}_{\forall t'\le 2t,  \mu_{t'}[\omega]\in B_{\delta_1}(\mu_0)\& \text{supp}(\mu_{2t+1}|\mu_{2t})\in B_{\delta_1}(\mu_0)} \Bigg]\\
			\ge&\E\Bigg[\sum_t \Bigg(\int (\mu_{2t+1}-\mu_{2t})^TB(\mu_0)(\mu_{2t+1}-\mu_{2t})\d \pi_{2t}(\mu_{2t+1}|\mu_{2t})-(\eta+\xi)\int \|\mu_{2t+1}-\mu_{2t}\|^2\d \pi_{2t}(\mu_{2t+1}|\mu_{2t})\Bigg)\\
					&\times\mathbf{1}_{\forall t'\le 2t,  \mu_{t'}[\omega]\in B_{\delta_1}(\mu_0)\& \text{supp}(\mu_{2t+1}|\mu_{2t})\in B_{\delta_1}(\mu_0)} \Bigg]
\end{align*}
The inequality is implied by the differentiability condition of $C$ and definition of continuity parameter $\xi$.\par
To calculate \cref{eqn:thm2:Ib}, we modify the distribution of $\mu_{2t+1}$ for any event that satisfies the restriction, namely $\mu_{2t}\in B_{\delta_1}(\mu_0)$ but there exists $\mu_{2t+1}\not\in B_{\delta_1}(\mu_0)$. For any such $\mu_{2t}$, let $\pi_{2t}$ still be the distribution of $\mu_{2t+1}$. Now let $\mu'=\E_{\pi_{2t}}[\nu|\nu\not\in B_{\delta_1}(\mu_0)]$. Let $\pi''(\nu|\mu')=\pi_{2t}(\nu|\nu\not\in B_{\delta_1}(\mu_0))$ and $\pi''(\nu|\mu)=\delta_{\nu}$ otherwise. Let $\pi'(\nu)=\mathbf{1}_{\nu\in B_{\delta_1}(\mu_0)}\pi_{2t}(\nu)+\pi_{2t}(\Delta(X)\setminus B_{\delta_1}(\mu_0))\delta_{\mu'}$. It is easy to verify that $\pi_{2t}(\nu)=\E_{\pi'}[\pi''(\nu|\mu)]$. Suppose $\mu'\not\in B_{\delta}(\mu_0)$, let $\mu''$ be a linear combination of $\mu_{2t}$ and $\mu'$ that is on the boundary of $B_{\delta}(\mu_0)$. Then we construct $\widetilde{\pi}'$ by shifting $\mu'$ to $\mu''$. Let $\mu'_1=\E_{\pi_{2t}}[\nu|\nu\in B_{\delta_1}(\mu_0)]$. Define:
\begin{align*}
	\widetilde{\pi}'(\nu)=\mathbf{1}_{\nu\in B_{\delta_1}(\mu_0)}\pi_{2t}(\nu)\cdot \frac{\|\mu'-\mu'_1\|}{\|\mu'-\mu_{2t}\|}\cdot \frac{\|\mu''-\mu_{2t}\|}{\|\mu''-\mu'_1\|}+\frac{\|\mu'_1-\mu_{2t}\|}{\|\mu''-\mu'_1\|}\delta_{\mu''}(\nu)
\end{align*}
By definition, $\E_{\widetilde{\pi}'}[\nu]=\mu_{2t}$ and $\widetilde{\pi}'\le_{BW} \pi'$. When $\mu'\in B_{\delta}(\mu_0)$, let $\widetilde{\pi}'=\pi'$. Now we calculate the cost of $\pi_{2t}$:
\begin{align*}
	C(\pi_{2t})\ge& C(\pi') \ge C(\widetilde{\pi}')\\
	\ge& \int (\nu-\mu_{2t})^TB(\mu_{0})(\nu-\mu_{2t})\d \widetilde{\pi}'(\nu)-(\eta+\xi)\int\|\nu-\mu_{2t}\|^2\d \widetilde{\pi}'(\nu)\\
	\ge& \int (\nu-\mu_{2t})^TB(\mu_{0})(\nu-\mu_{2t})\d \widetilde{\pi}'(\nu)-(\eta+\xi)\int\|\nu-\mu_{2t}\|^2\d \pi_{2t}(\nu)\\
	=&\int (\nu-\mu_{2t})^TB(\mu_{0})(\nu-\mu_{2t})\d \pi_{2t}(\nu)+\int (\nu-\mu_{2t})^TB(\mu_{0})(\nu-\mu_{2t})\d (\widetilde{\pi}'-\pi_{2t})(\nu)\\
	 &-(\eta+\xi)\int\|\nu-\mu_{2t}\|^2\d \pi_{2t}(\nu)\\
	=&\int (\nu-\mu_{2t})^TB(\mu_{0})(\nu-\mu_{2t})\d \pi_{2t}(\nu)-(\eta+\xi)\int\|\nu-\mu_{2t}\|^2\d \pi_{2t}(\nu)\\
	 &+\int_{\nu\in B_{\delta}(\mu_0)} (\nu-\mu_{2t})^TB(\mu_{0})(\nu-\mu_{2t})\d (\widetilde{\pi}'-\pi_{2t})(\nu)\\
	 &+\int_{\nu\not\in B_{\delta}(\mu_0)} (\nu-\mu_{2t})^TB(\mu_{0})(\nu-\mu_{2t})\d (\widetilde{\pi}'-\pi_{2t})(\nu)\\
	\ge&\int (\nu-\mu_{2t})^TB(\mu_{0})(\nu-\mu_{2t})\d \pi_{2t}(\nu)-(\eta+\xi)\int\|\nu-\mu_{2t}\|^2\d \pi_{2t}(\nu)\\
		 &-\left( 1-\frac{\delta-\delta_1}{\delta+\delta_1} \right)(\nu-\mu_{2t})^TB(\mu_{0})(\nu-\mu_{2t})\d \pi_{2t}(\nu)\\
		 &-\pi_{2t}(\Delta(X)\setminus B_{\delta_1}(\mu_0))\left(\int (\nu-\mu_{2t})B(\mu_0)(\nu-\mu_{2t})\d \pi''(\nu|\mu')+ (\mu'-\mu_{2t})^TB(\mu_0)(\mu'-\mu_{2t})\right)
\end{align*}
The first two inequalities are by $\pi_{2t}\ge_{BW}\pi'\ge_{BW}\widetilde{\pi}'$ and \cref{axiom:1}. The third inequality is by \cref{ass:1}. The forth inequality: $\frac{\|\mu'-\mu'_1\|}{\|\mu'-\mu_{2t}\|}<1$; $\frac{\|\mu''-\mu_{2t}\|}{\|\mu''-\mu'_1\|}\ge \frac{\delta-\delta_1}{\delta+\delta_1}$ bounds the second line, the third line is calculated by ignoring the weakly positive term provided by $\widetilde{\pi}'$. Finally, since $B(\mu_0)$ is fixed, there exists $M=\sup_{\mu,\nu\in \Delta(X)} (\nu-\mu)^TB(\mu_0)(\nu-\mu)<\infty$.  To sum up:
\begin{align}
	C(\pi_{2t})\ge& \int (\nu-\mu_{2t})^TB(\mu_{0})(\nu-\mu_{2t})\d \pi_{2t}(\nu)-\left(\eta+\xi+\frac{2\delta_1}{\delta+\delta_1}\right)\int\|\nu-\mu_{2t}\|^2\d \pi_{2t}(\nu)	\label{eqn:thm2:first}\\
	&-\pi_{2t}(\Delta(X)\setminus B_{\delta_1}(\mu_0))M   \notag
\end{align}
Plug \cref{eqn:thm2:first} into \cref{eqn:thm2:Ib}, we get:
\begin{align*}
	(\ref{eqn:thm2:Ib})\ge&\E\Bigg[\sum_t \Bigg(\int (\mu_{2t+1}-\mu_{2t})^TB(\mu_0)(\mu_{2t+1}-\mu_{2t})\d \pi_{2t}(\mu_{2t+1}|\mu_{2t})\\
										 &-\left(\eta+\xi+\frac{2\delta_1}{\delta+\delta_1}\right)\int \|\mu_{2t+1}-\mu_{2t}\|^2\d \pi_{2t}(\mu_{2t+1}|\mu_{2t})+\pi_{2t}\left( \Delta(X)\setminus B_{\delta_1}(\mu_0) |\mu_{2t}\right)\cdot M\Bigg)\\
					&\times\mathbf{1}_{\forall t'\le 2t,  \mu_{t'}[\omega]\in B_{\delta_1}(\mu_0)\& \text{supp}(\mu_{2t+1}|\mu_{2t})\not\in B_{\delta_1}(\mu_0)} \Bigg] \\
					=&\E\Bigg[\sum_t \Bigg(\int (\mu_{2t+1}-\mu_{2t})^TB(\mu_0)(\mu_{2t+1}-\mu_{2t})\d \pi_{2t}(\mu_{2t+1}|\mu_{2t})\\
										 &-\left(\eta+\xi+\frac{2\delta_1}{\delta+\delta_1}\right)\int \|\mu_{2t+1}-\mu_{2t}\|^2\d \pi_{2t}(\mu_{2t+1}|\mu_{2t})\Bigg)\\
					&\times\mathbf{1}_{\forall t'\le 2t,  \mu_{t'}[\omega]\in B_{\delta_1}(\mu_0)\& \text{supp}(\mu_{2t+1}|\mu_{2t})\not\in B_{\delta_1}(\mu_0)} \Bigg]\\
					&-P(\cup \Omega_0^{\alpha})\cdot M\\
							\ge&\E\Bigg[\sum_t \Bigg(\int (\mu_{2t+1}-\mu_{2t})^TB(\mu_0)(\mu_{2t+1}-\mu_{2t})\d \pi_{2t}(\mu_{2t+1}|\mu_{2t})\\
										 &-\left(\eta+\xi+\frac{2\delta_1}{\delta+\delta_1}\right)\int \|\mu_{2t+1}-\mu_{2t}\|^2\d \pi_{2t}(\mu_{2t+1}|\mu_{2t})\Bigg)\\
										 &\times\mathbf{1}_{\forall t'\le 2t,  \mu_{t'}[\omega]\in B_{\delta_1}(\mu_0)\& \text{supp}(\mu_{2t+1}|\mu_{2t})\not\in B_{\delta_1}(\mu_0)} \Bigg]]\\
					&-\frac{3\epsilon}{(\delta_1^2+\delta_0^2)m}
\end{align*}
The equality is by rewriting the event space at which some path first crosses $B_{\delta_1}(\mu_0)$, The last inequality is implied by \cref{eqn:thm2:first}.

Now we study \cref{eqn:thm2:II,eqn:thm2:III}:
\begin{align*}
	&(\ref{eqn:thm2:II})+(\ref{eqn:thm2:III})\\
	\ge&\E\left[\sum_t \left(\int (\mu_{2t+1}-\mu_{2t})^TB(\mu_0)(\mu_{2t+1}-\mu_{2t})\d \pi_{2t}(\mu_{2t+1}|\mu_{2t})\right)\mathbf{1}_{\exists t'\le 2t, \mu_{t'}[\omega]\not\in B_{\delta_1}(\mu_0)}\right]\\
	 &-\E\left[\sum_t \left(\int (\mu_{2t+1}-\mu_{2t})^TB(\mu_0)(\mu_{2t+1}-\mu_{2t})\d \pi_{2t}(\mu_{2t+1}|\mu_{2t})\right)\mathbf{1}_{\exists t'\le 2t, \mu_{t'}[\omega]\not\in B_{\delta_1}(\mu_0)}\right]\\
	\ge&\E\left[\sum_t \left(\int (\mu_{2t+1}-\mu_{2t})^TB(\mu_0)(\mu_{2t+1}-\mu_{2t})\d \pi_{2t}(\mu_{2t+1}|\mu_{2t})\right)\mathbf{1}_{\exists t'\le 2t, \mu_{t'}[\omega]\not\in B_{\delta_1}(\mu_0)}\right]\\
		 &-\E\left[ \sum_t\left( \sum_{s\ge t}\int(\mu_{2s+1}-\mu_{2t})^TB(\mu_0)(\mu_{2s+1}-\mu_{2t})\d \pi_{2t}(\mu_{2s+1}|\mu_{2s}) \right)\mathbf{1}_{\forall t'<2t, \mu_{t'}[\omega]\in B_{\delta_1}(\mu_0)\& \mu_{2t}[\omega]\not\in B_{\delta_1}(\mu_0)} \right]\\
	=&\E\left[\sum_t \left(\int (\mu_{2t+1}-\mu_{2t})^TB(\mu_0)(\mu_{2t+1}-\mu_{2t})\d \pi_{2t}(\mu_{2t+1}|\mu_{2t})\right)\mathbf{1}_{\exists t'\le 2t, \mu_{t'}[\omega]\not\in B_{\delta_1}(\mu_0)}\right]\\
	 &-\E\left[  \sum_{2s\ge t_0(\alpha)}\int(\mu_{2s+1}-\mu_{2t})^TB(\mu_0)(\mu_{2s+1}-\mu_{2t})\d \pi_{2t}(\mu_{2s+1}|\mu_{2s})\Big|\cup \Omega_0^{\alpha} \right]P(\cup\Omega_0^{\alpha})\\
	\ge&\E\left[\sum_t \left(\int (\mu_{2t+1}-\mu_{2t})^TB(\mu_0)(\mu_{2t+1}-\mu_{2t})\d \pi_{2t}(\mu_{2t+1}|\mu_{2t})\right)\mathbf{1}_{\exists t'\le 2t, \mu_{t'}[\omega]\not\in B_{\delta_1}(\mu_0)}\right]\\
		 &-\|B(\mu_0)\|\E\left[  \sum_{2s\ge t_0(\alpha)}\int\|\mu_{2s+1}-\mu_{2s}\|^2\d \pi_{2t}(\mu_{2s+1}|\mu_{2s})\Big|\cup \Omega_0^{\alpha} \right]P(\cup\Omega_0^{\alpha})\\
	\ge&\E\left[\sum_t \left(\int (\mu_{2t+1}-\mu_{2t})^TB(\mu_0)(\mu_{2t+1}-\mu_{2t})\d \pi_{2t}(\mu_{2t+1}|\mu_{2t})\right)\mathbf{1}_{\exists t'\le 2t, \mu_{t'}[\omega]\not\in B_{\delta_1}(\mu_0)}\right]\\
		 &-\|B(\mu_0)\|\frac{3\epsilon}{m}
\end{align*}
The first inequality is by definition of $M$. The second equality is by definition any path which ever crosses $B_{\delta_1}(\mu_{0})$ must have first crossed it at some history. The equality is rewriting the second term in the language of $\Omega_0^{\alpha}$ and $t_0(\alpha)$. The last inequality is implied by \cref{eqn:thm2:2}.\par

Now combine \cref{eqn:thm2:I,eqn:thm2:II,eqn:thm2:III} together and we get:
\begin{align*}
	\E[\sum_t C(\mu_{2t+1}|\mu_{2t})]\ge&\E\left[ \sum_t \int(\mu_{2t+1}-\mu_{2t})^TB(\mu_0)(\mu_{2t+1}-\mu_{2t})\d \mu_{2t}(\mu_{2t+1}|\mu_{2t}) \right]\\
																			&-\left( \eta+\xi+\frac{2\delta_1}{\delta+\delta_1} \right)\E\left[ \sum_t \int \|\mu_{2t+1}-\mu_{2t}\|^2\d\pi_{2t}(\mu_{2t+1}|\mu_{2t}) \right]\\
																			&-\left( \frac{1}{\delta_1^2+\delta_0^2}+\|B(\mu_0)\| \right)\frac{3\epsilon}{m}\\
	\ge& \int (\nu-\mu_0)^TB(\mu_0)(\nu-\mu_0)\d \pi(\nu)\\
		 &-\left( \eta+\xi+\frac{2\delta_1}{\delta+\delta_1} \right)\E\left[ \sum_t \int \|\mu_{2t+1}-\mu_{2t}\|^2\d\pi_{2t}(\mu_{2t+1}|\mu_{2t}) \right]\\
																			&-\left( \frac{1}{\delta_1^2+\delta_0^2}+\|B(\mu_0)\| \right)\frac{3\epsilon}{m}\\
	\ge& \int (\nu-\mu_0)^TB(\mu_0)(\nu-\mu_0)\d \pi(\nu)\\
	\ge& \int (\nu-\mu_0)^TB(\mu_0)(\nu-\mu_0)\d \pi(\nu)\\
		 &-\left( \eta+\xi+\frac{2\delta_1}{\delta+\delta_1} \right)\left( \int\|\nu-\mu_0\|^2 \d \pi(\nu)+\frac{\epsilon}{m} \right) \\
																			&-\left( \frac{1}{\delta_1^2+\delta_0^2}+\|B(\mu_0)\| \right)\frac{3\epsilon}{m}
\end{align*}
Fix all other parameters and let $\epsilon\to0$ ($\epsilon$ is the approximation error defined in \cref{eqn:1}), then this implies:
\begin{align*}
	C^*(\pi)\ge \int (\nu-\mu_0)^TB(\mu_0)(\nu-\mu_0)\d \pi(\nu) -\left( \eta+\xi+\frac{2\delta_1}{\delta+\delta_1} \right)\left( \int\|\nu-\mu_0\|^2 \d \pi(\nu) \right)
\end{align*}\par

Therefore, $\forall \epsilon>0$, let $\delta$ be its corresponding parameter define in \cref{ass:1}, we can pick $\delta_0,\delta_1$ small enough such that $\eta+\xi+\frac{2\delta_1}{\delta+\delta_1}<\epsilon$.\footnote{Here we recycled symbol $\epsilon$, now it is used to denote the parameter defining the differentiability condition of $C^*$.} Hence we proved that we find $\delta_0$ s.t. $\forall \pi$ s.t. $\E_{\pi}[\nu]=\mu_0$ and supp$(\pi)\subset B_{\delta_0}(\mu_0)$:
\begin{align*}
	C^*(\pi)\ge \int (\nu-\mu_0)^TB(\mu_0)(\nu-\mu_0)\d \pi(\nu) -\epsilon\left( \int\|\nu-\mu_0\|^2 \d \pi(\nu) \right)
\end{align*}
On the other hand:
\begin{align*}
	C^*(\pi)\le C(\pi)\le \int (\nu-\mu_0)^TB(\mu_0)(\nu-\mu_0)\d \pi(\nu) +\epsilon\left( \int\|\nu-\mu_0\|^2 \d \pi(\nu) \right)
\end{align*}
Therefore, 
\begin{align*}
	\left|C^*(\pi)- \int (\nu-\mu_0)^TB(\mu_0)(\nu-\mu_0)\d \pi(\nu)\right|\le \epsilon\left( \int\|\nu-\mu_0\|^2 \d \pi(\nu) \right)
\end{align*}
We verified the twice differentiability of $C^*(\pi)$ and \cref{eqn:thm2:3}. By definition, if $C^*(\pi)=\E_{\pi}[H(\nu)]-H(\E_{\pi}[\nu])$, then $2B(\mu)\equiv \mathbb{H}H(\mu)$ is the Hessian matrix of $H(\mu)$. In other words, $B(\mu)$ also locally characterizes $C^*(\mu)$. \par

Therefore:
\begin{align*}
	C(\pi)\ge C^*(\pi)=\E_{\pi}[H(\nu)]-H(\E_{\pi}[\nu])
\end{align*}\par

\textbf{Sufficiency}: Second, we prove that $C$ favoring incremental evidences implies $\phi(C)$ being uniformly posterior separable. Suppose $B(\mu)$ locally characterizes $C$ and $C$ favors incremental evidences. Define:
\begin{align*}
	\underline{C}(\pi)=\E_{\pi}[H(\nu)]-H(\E_{\pi}[\nu])\le C(\pi)
\end{align*}
where convex function $H(\mu)$ has Hessian matrix $2B(\mu)$. $\underline{C}$ satisfies \cref{axiom:0,axiom:1,axiom:2}, therefore by \cref{thm:1}, $\phi(\underline{C})=\underline{C}$. Therefore it is sufficient to show that $\phi(C)\le \underline{C}$.\footnote{Although not formally shown, it is straight forward that $\phi$ is an increasing map.} \par

We first show that when $\text{supp}(\pi)\subset \Delta(X)^o$ and is \emph{binary}, $\phi(C)(\pi)\le \underline{C}(\pi)$. We prove by finding $\pi'$ arbitrarily close to $\pi$ under $L-P$ metric and $\langle \mu_t \rangle\rep \pi'$ with cost arbitrarily close to $\underline{C}$. Since $|\text{supp}(\pi)|=2$, denote the two posterior beliefs by $\nu_1,\nu_2$. Pick $M\in\mathbb{N}$, $\forall i\in\mathbb{N}$, $i\le M$ define $\lambda_i=\frac{i}{M}$. Consider the subspace $\left\{ \mu_i=\lambda_i\nu_1+(1-\lambda_i)\nu_2 \right\}$. Let $\mu_{m_0}$ be the closest point to $\E_{\pi}[\nu]$. Define information structure $\widehat{\pi}$ to be with prior $\mu_{m_0}$ and posteriors $\mu_0,\mu_M$. Then $\lim_{M\to \infty}d(\pi,\widehat{\pi})_{lp}=0$  By continuity of $\phi(C)$ and $\underline{C}$, $\forall \epsilon>0$, $\exists M$ large enough that $|\phi(C)(\pi)-\phi(C)(\widehat{\pi})|\le \epsilon$ and $|\underline{C}(\pi)-\underline{C}(\widehat{\pi})|\le \epsilon$. \par

Now consider the following process $\langle \mu_t \rangle$, defined as follows: $\mu_0=\mu_{m_0}$, $\pi_t(\mu_{t+1}|\mu_{i})=\frac{1}{2}\delta_{\mu_{i+1}}+\frac{1}{2}\delta_{\mu_{i-1}}$ when $i\in[1,M-1]$. $\pi_t(\mu_{t+1}|\mu_0)=\delta_{\mu_0}$ and $\pi_t(\mu_{t+1}|\mu_M)=\delta_{\mu_M}$. In other words, $\langle \mu_t \rangle$ is a standard random walk in $\left\{ \mu_i \right\}$ stopped at absorbing boundary $\left\{ \mu_0,\mu_M \right\}$. Let $T\in\mathbb{N}$ be the length of the process $\langle \mu_t \rangle$. Then it is easy to verify that $\text{prob}(\mu_T\in\left\{ \mu_0,\mu_M \right\})\to 1$ when $T\to \infty$.\footnote{Let $P_T$ be such probability, then i) $P_T$ is increasing, ii) $P_{T+M}\ge P_T+(1-P_T)\frac{1}{2^M}$. } Therefore, $\exists T$ large enough, s.t. if we let $\mu_{T}\sim \pi'$, then $|\phi(C)(\pi')-\phi(C)(\widehat{\pi})|\le \epsilon$ and $|\underline{C}(\pi')-\underline{C}(\widehat{\pi})|\le \epsilon$. By definition, finite process $\langle \mu_t \rangle\rep \pi'$. Notice that $\langle \mu_t \rangle$ does not have the information disposal periods. The cost of $\langle \mu_t \rangle$ is:
\begin{align}
	\sum C(\pi_t(\mu_{t+1}|\mu_t))\le & \E\left[ \sum (H(\mu_{t+1})-H(\mu_t)) \right]\label{eqn:suf:1}\\
																		&+ \E\left[ \sum \left|(H(\mu_{t+1})-H(\mu_t))-\frac{1}{M^2}(\nu_2-\nu_1)^TB(\mu_t)(\nu_2-\nu_1)\right| \right]\label{eqn:suf:2}\\
																		&+\E\left[ \sum\left|\frac{1}{M^2}(\nu_2-\nu_1)^TB(\mu_t)(\nu_2-\nu_1)-C(\pi_t(\mu_{t+1}|\mu_t))\right| \right]\label{eqn:suf:3}
\end{align}
The first term (\ref{eqn:suf:1}) is exactly $\E_{\pi'}[H(\nu)]-H(\E_{\pi'}[\nu])=\underline{C}(\pi')$. Now consider the second term (\ref{eqn:suf:2}). $\forall \mu_{i}$, let $f(\alpha)=H(\mu_i+\frac{\alpha}{M}(\nu_2-\nu_1))$. Then:
\begin{align*}
	&\frac{1}{2}\inf_{\alpha\in[0,1]} f''(\alpha)\le f(1)-f(0)-f'(0)\le \frac{1}{2}\sup_{\alpha\in[0,1]}f''(\alpha)\\
	\iff& \inf_{\alpha\in[0,1]}\frac{1}{2}\frac{1}{M^2}(\nu_2-\nu_1)^T\mathbb{H}H(\mu_i+\frac{\alpha}{M}(\nu_1-\nu_1))(\nu_2-\nu_1)\\
			&\le H(\mu_{i+1})-H(\mu_i)-\frac{1}{M}\mathbb{J}H(\mu_i)(\nu_2-\nu_1)\\
			&\le\sup_{\alpha\in[0,1]}\frac{1}{2}\frac{1}{M^2}(\nu_2-\nu_1)^T\mathbb{H}H(\mu_i+\frac{\alpha}{M}(\nu_1-\nu_1))(\nu_2-\nu_1)\\
		\implies& \left| H(\mu_{i+1})-H(\mu_i)-\frac{1}{M}\mathbb{J}H(\mu_i)(\nu_2-\nu_1) -\frac{1}{M^2}(\nu_2-\nu_1)^TB(\mu_i)(\nu_2-\nu_1) \right|\\
						&\le \frac{1}{M^2}\|\nu_2-\nu_1\|^2\cdot \sup_{\mu'\in [\mu_{i-1},\mu_{i+1}]} \|B(\mu')-B(\mu)\|
\end{align*}
Since by \cref{lem:diff} $B(\mu)$ is continuous on $\alpha(X)^o$, thus is uniformly continuous on $ [\mu_0,\mu_m] $. Therefore, $M$ can be picked large enough that $\sup_{\mu'\in [\mu_{i-1},\mu_{i+1}]} \|B(\mu')-B(\mu)\|\le \epsilon$. This implies
\begin{align*}
	(\ref{eqn:suf:2})\le \epsilon\cdot \E\left[ \sum \|\mu_{t+1}-\mu_t\|^2 \right]=\epsilon \cdot\E_{\pi'}[\|\nu-\mu_{m_0}\|^2]
\end{align*}

Now consider (\ref{eqn:suf:3}). Since $[\mu_0,\mu_{M}]$ is compact, there exists uniform $\delta$ satisfying \cref{lem:diff}. Therefore, when $M$ is picked larger than $\frac{1}{\delta}$, $\mu_{t+1}$ is always within $B_{\delta}(\mu_t)$ and:
\begin{align*}
	(\ref{eqn:suf:3})\le &\epsilon\cdot \E\left[ \sum \|\mu_{t+1}-\mu_t\|^2 \right]=\epsilon \cdot\E_{\pi'}[\|\nu-\mu_{m_0}\|^2]
\end{align*}
To sum up:
\begin{align*}
	\sum C(\pi_t(\mu_{t+1}|\mu_t))\le& (\ref{eqn:suf:1})+(\ref{eqn:suf:2})+(\ref{eqn:suf:3})\\
	\le & \underline{C}(\pi')+2\epsilon \cdot\E_{\pi'}[\|\nu-\mu_{m_0}\|^2]\\
	\implies \phi(C)(\pi')\le& \underline{C}(\pi')+2\epsilon \cdot\E_{\pi'}[\|\nu-\mu_{m_0}\|^2]\\
	\implies \phi(C)(\pi)\le& \underline{C}(\pi')+2\epsilon \cdot\E_{\pi'}[\|\nu-\mu_{m_0}\|^2]+2\epsilon\\
	\le& \underline{C}(\pi)+2\epsilon \cdot\E_{\pi'}[\|\nu-\mu_{m_0}\|^2]+4\epsilon
\end{align*}
Since $\epsilon$ can be arbitrarily small, $\phi(C)(\pi)\le \underline{C}(\pi)$. Therefore, $\phi(C)(\pi)=\underline{C}(\pi)$. Finally, continuity of $\underline{C}(\pi)$ and $\phi(C)(\pi)$ on $I(X)$ implies that the equality extends to any binary support $\pi$ in $\mathcal{I}(X)$. Since the argument also applies to any $\mathcal{I}(X')$, we can claim that $\phi(C)\equiv \underline{C}$ for information structures with binary support. \par

Now we prove the statement for general $\pi\in \mathcal{I}(X)$. $\forall \epsilon>0$, let $\delta$ be the continuity parameter of $\phi(C)$ and $\underline{C}$ at $\pi$. First consider any finite (Borel) partition $\cup_{i=1}^{M} D_i$ of $\Delta(X)$ where the diameter of any $D_i$ is bounded by $\delta$. Wlog, we consider the case $\pi(D_i)>0$ and $M\ge 3$. Let $\mu_i=\E_{\pi}[\nu|\nu\in D_i]$. Define $\pi'(\mu)=\sum_i\pi(D_i)\cdot \delta_{\mu_i}(\mu)$. Then $d(\pi,\pi')_{lp}\le \delta$ and hence $|\phi(C)(\pi)-\phi(C)(\pi')|\le \epsilon$. Let $\mu_0=\E_{\pi}[\nu]$. Now we consider a decomposition of $\pi'$:
\begin{align*}
	\begin{dcases}
		\pi_1=\pi(D_1)\delta_{\mu_1}+(1-\pi(D_1))\delta_{\E_{\pi}[\nu|\nu\not\in D_1]}\\
		\pi_i(\cdot|\mu)=\frac{\pi(D_i)}{\sum_{j\ge i}\pi(D_j)}\delta_{\mu_i}(\cdot)+\frac{\sum_{j>i}\pi(D_j)}{\sum_{j\ge i}\pi(D_j)}\delta_{\E_{\pi}[\nu|\nu\in \cup_{j>i}D_i]}(\cdot)&\text{when }\mu=\E_{\pi}[\nu|\nu\in \cup_{j> i}D_j]\\
		\pi_i(\cdot|\mu)=\mu&\text{otherwise}
	\end{dcases}
\end{align*}
By definition $\pi'(\nu)=\E[\prod_{i=1}^{M-1} \pi_i]$. Therefore by recursively applying \cref{axiom:2}:
\begin{align*}
	\phi(C)(\pi')\le& \E\left[ \sum \phi(C)(\pi_i) \right]\\
	=&\pi(D_1)H(\mu_1)+\sum_{j>1}\pi(D_j)H(\E_{\pi}[\nu|\nu\in\cup_{j>1}D_j])-H(\mu_0)\\
	 &+\left( \sum_{j>1}\pi(D_j) \right)\left( \frac{\pi(D_2)}{\sum_{j>1}\pi(D_j)}H(\mu_2)+\frac{\sum_{j>2}\pi(D_j)}{\sum_{j>1}\pi(D_j)}H(\E_{\pi}[\nu|\nu\in \cup_{j>2}D_j])-H(\E_{\pi}[\nu|\nu\in\cup_{j>1}D_j]) \right)\\
	 &+\cdots\\
	 &+\left( \sum_{j>1}\pi(D_j) \right)\prod_{i=1}^{M-2} \frac{\sum_{j>i}\pi(D_j)}{\sum_{j\ge i}\pi(D_j)}\Bigg( \frac{\pi(D_{M-1})}{\pi(D_{M-1})+\pi(D_{M})}H(\mu_{M-1})+\frac{\pi(D_M)}{\pi(D_{M-1})+\pi(D_M)}H(\mu_M)\\
	 &-H(\E_{\pi}[\nu|\nu\in D_{M-1}\cup D_M])\Bigg)\\
	=&\sum \pi(D_i)H(\mu_i)-H(\mu_0)\\
	=&\underline{C}(\pi')
\end{align*}
The first equality utilizes the result with binary support information structures. The second equality is from cancelling out terms. Since $\epsilon$ can be chosen arbitrarily, $\phi(C)(\pi)\le\underline{C}(\pi)$. 
\end{proof}

\subsection{Proof of \cref{prop:dynamic:solution}}

\begin{proof}
	The second and third properties are straight forward. One only need to observe that the weak inequality becomes strict in the proof of \cref{prop:dynamic} once the uniqueness is violated. \par
	We state the detailed proof for the first property. Suppose $C$ satisfies strict-monotonicity. Let $m=|A|$. For this proof we take a direct approach to model information structure and re-parametrize $C(\cdot)$ using $\widehat{C}_m$. Consider two conditional distributions of actions $P_1\neq P_2$. Let $P=\lambda P_1+(1-\lambda)P_2$. We want to show that strict monotonicity implies that the cost of $P$ being strictly lower than the convex combination of costs of $P-1$ and $P_2$.\par
	Define $P'$ on $A\times\left\{ 1,2 \right\}=\left\{ a_1,a_2,a'_1,a'_2,\dots \right\}$ with twice number of signals than $A$. Let $\lambda_1=\lambda,\lambda_2=1-\lambda$, $\forall a,x$, define
	\begin{align*}
		P'(a_i|x)=\lambda_iP_i(a|x)
	\end{align*}
	Then $P'\ge_{BW}P$:
	\begin{align}
		\label{eqn:mtx}
		\begin{bmatrix}
			1\ 1&0\ 0&\cdots&0\ 0\\
			0\ 0&1\ 1&\cdots&0\ 0\\
			\vdots&\vdots&\ddots&\vdots\\
			0\ 0&0\ 0&\cdots&1\ 1
		\end{bmatrix}
		\cdot P'=P
	\end{align}\par
	and $\widehat{C}_m(P,\mu)\le\widehat{C}_{2m}(P',\mu)\le\lambda\widehat{C}_m(P_1,\mu)+(1-\lambda)\widehat{C}_m(P_2,\mu)$. Suppose for the sake of contradiction that equality $\widehat{C}_m(P,\mu)=\widehat{C}_{2m}(P',\mu)$ holds, then strict-monotonicity implies that $P\ge_{BW}P'$:
	\begin{align*}
		M\cdot P=P'
	\end{align*}
	Where $M$ is a stochastic matrix. Consider the following operation one the rows of $P'$: If $P_1\sim P_2$, then proof is done. Otherwise, first remove replication of $P'$ (when two rows of $P'$ are multiplications of each other, add them up) and get $\widetilde{P}'$. Since $P_1\not\sim P_2$, we can assume $\widetilde{P}'{}^1=P'_1{}^1,\widetilde{P}'{}^2=P'_2{}^1$. Define $\widehat{P}^1=P'{}^1+P'{}^2$ and $\widehat{P}^i=\widetilde{P}'{}^{i+1}$. By definition $\widetilde{P}'$ Blackwell dominates $\widehat{P}$. On the other hand, $\widehat{P}$ Blackwell dominates $P$, so dominates $P'$, and $\widetilde{P}'$. By \cref{lem:BWequiv}, $\widetilde{P}'$ and $\widehat{P}$ are identical up to permutation. Then $P'{}^1$ must equal to some $\widehat{P}^i$.
	\begin{itemize}
		\item \emph{Case 1}. If $i=1$, then $P'{}^1+P'{}^2$ is a multiplication of $\widetilde{P}'_1{}^1$, which is a multiplication of $\widetilde{P}'{}^1$. This means $P'{}^1$ and $P'{}^2$ are replication, contradiction.
		\item \emph{Case 2}. If $i>1$, then $\widetilde{P}'{}^1$ is a multiplication of $\widehat{P}^i$, which is a multiplication of $\widetilde{P}'{}^{i+1}$. Contradicting definition of $\widetilde{P}'$ (no replication).
	\end{itemize}
	Therefore, $P'{}^1$ and $P'{}^2$ are replications. Now permute $P'$ and apply the same analysis on all $P'{}^{2i-1},P'{}^{2i}$. We can conclude that any row of $P_1$ is a replication of that of $P_2$. To sum up, a necessary condition for $\widehat{C}_m(P,\mu)=\lambda \widehat{C}_m(P_1,\mu)+(1-\lambda)\widehat{C}_m(P_2,\mu)$ is that each row in $P_1$ and $P_2$ induces same posterior belief $\nu$.\par
	Now consider the set of solutions to \cref{eqn:dynamic:static}. Strict monotonicity implies that the number of signal realizations must be no higher than $m$ and wlog each $\pi$ can be represented as $P\in\mathbb{R}^{m\times n}$. Suppose by contradiction there exists $P_1$ and $P_2$ and $a$ such that they induces different posterior for action $a$. Consider any $P=\lambda p_1+(1-\lambda) P_2$. By previous proof, $\widehat{C}_m(P,\mu)<\lambda\widehat{C}_m(P_1,\mu)+(1-\lambda)\widehat{C}_m(P_2,\mu)$. Namely we find a strict improvement over $P_1$ and $P_2$. Contradiction.\par
To sum up, solutions to \cref{eqn:dynamic:static} always have the same support. 

\end{proof}
\begin{lem}[Blackwell equivalence]
	\label{lem:BWequiv}
	Let $P$ and $P'$ be two stochastic matrices. $P$ has no replication of rows. Suppose there exists stochastic matrices $M_{PP'}$ and $M_{P'P}$ s.t.:
	\begin{align*}
		P'&=M_{P'P}\cdot P\\
		P&=M_{PP'}\cdot P'
	\end{align*}
	Then $M_{PP'}$ and $M_{P'P}$ are permutation matrices.
\end{lem}
\begin{proof}
	Let $P_i=\left( p_{i1},p_{i2},\ldots \right)$ be $i$th row of $P$. Suppose $P_i$ can not be represented as positive combination of $P_{-i}$'s. Then by construction $P_i=M_{PP'i}\cdot M_{P'P}\cdot P$, we have:
	\begin{align*}
		M_{PP'i}\cdot M_{P'P}=( \underbrace{0,\cdots,0}_{i-1},1,0,\cdots,0 )
	\end{align*}
	Then by non-negativity of stochastic matrices, suppose $M_{PP'ij}>0$, then $M_{P'Pj}$ are all $0$ except $M_{P'Pji}$. Then for all such rows $j$, we have $M_{PP'j}$ be a vector with only $i$th column being non-zero. However this suggests they are replicated rows. So the only possibility is that $j$ s.t. $M_{PP'ij}>0$ is unique. And
	\begin{align*}
		M_{PP'ij}\times M_{P'Pji}=1
	\end{align*}
	Since stochastic matrices have elements no larger than $1$, it must be $M_{PP'ij}=M_{P'Pji}=1$. This is equivalently saying $P'_j=P_i$. Since permutation of rows of $P'$ doesn't affect our statement, let's assume $P'_i=P_i$ afterwards for simplicity.\par
	So far we showed that if $P_i$ is not a positive combinations of $P_{-i}$'s, then $P'_i=P_i$. We do the following transformation: $\widetilde{P},\widetilde{P}'$ are $P,P'$ removing $i$th row. $\widetilde{M}_{PP'},\widetilde{M}_{P'P}$ are $M_{PP'},M_{P'P}$ removing $i$th row and column. It's easy to verify that we still have:
	\begin{align*}
		\widetilde{P}'&=\widetilde{M}_{P'P}\cdot\widetilde{P}\\
		\widetilde{P}&=\widetilde{M}_{PP'}\cdot\widetilde{P}'
	\end{align*}
	and $\widetilde{M}_{PP'},\widetilde{M}_{P'P}$ still being stochastic matrices since previous argument shows $M_{PP'ii}$ and $M_{P'Pii}$ being the only non-zero element in their rows.  Since they are both $1$, they must also be only non-zero element in their columns. So removing them doesn't affect the matrices being stochastic matrices.\par
   	Now we can repeat this process iteratively until any row $\widetilde{P}_i$ will be a positive combination of $\widetilde{P}_{-i}$. If $\widetilde{P}$ has one unique row, then the proof is done. We essentially showed that $P=P'$ (up to permutation of rows). Therefore we only need to exclude the possibility of $\widetilde{P}$ having more than one rows.\par
	Suppose $\widetilde{P}$ has $n$ rows. Then $\widetilde{P}_1$ is a positive combination of $\widetilde{P}_{-i}$'s:
	\begin{align*}
		\widetilde{P}_1=\sum_{i=2}^{n} a^1_i\widetilde{P}_i
	\end{align*}
	and $\widetilde{P}_2$ is a positive combination of $\widetilde{P}_{-i}$'s:
	\begin{align*}
		\widetilde{P}_2=&\sum_{i\neq 2}^{n}a^2_i\widetilde{P}_i\\
		=&a^2_1\widetilde{P}_1+\sum_{i> 2}^{n}a^2_i\widetilde{P}_i\\
		=&a^2_1a^1_2\widetilde{P}_2+\sum_{i>2}^n\left( a^2_i+a^2_1a^1_i \right)\widetilde{P}_i
	\end{align*}
	Since all rows in $\widetilde{P}$ are non-negative (and strictly positive in some elements). This is possible only in two cases:
	\begin{itemize}
		\item \emph{Case 1}. $a^2_1a^1_2=1$ and $\sum_{i>2}\left( a^2_i+a^2_1a_i^1 \right)=0$. This implies $\widetilde{P}_1=a^{1}_2\widetilde{P}_2$. Contradicting non-replication.
		\item \emph{Case 2}. $a^2_1a^1_2<1$. Then $\widetilde{P}_2$ is a positive combination of $\widetilde{P}_{i>2}$. Of course $\widetilde{P}_1$ is also a positive combination of $\widetilde{P}_{i>2}$.
	\end{itemize}
	Now by induction suppose $\widetilde{P}_1,\ldots,\widetilde{P}_i$ are positive combinations of $\widetilde{P}_{j>i}$. Then:
	\begin{align*}
		\widetilde{P}_{i+1}=&\sum_{j=1}^i a^{i+1}_j \widetilde{P}_j+\sum_{j=i+1}^na^{i+2}_j\widetilde{P}_j\\
		&=\sum_{k=i}^{n}\left(\sum_{j=1}^i a^{i+1}_j a^j_k\right)\widetilde{P}_k+\sum_{j=i+2}^n\widetilde{P}_j\\
		&=\sum_{j=1}^ia^{i+1}_ja^j_{i+1}\widetilde{P}_{i+1}+\sum_{k=i+2}^{n}\left( \sum_{j=1}^ia_j^{i+1}a_k^j+a^{i+1}_j \right)\widetilde{P}_j
	\end{align*}
	Similar to previous analysis, non-replication implies $\sum_{j=1}^ia_j^{i+1}<1$ and $\widetilde{P}_{i+1}$ is a positive combination of $\widetilde{P}_{j>i+1}$. Then by replacing $\widetilde{P}_{i+1}$ in combination of all $\widetilde{P}_{j\le i}$, we can conclude that $\widetilde{P}_1,\ldots,\widetilde{P}_{i+1}$ are all positive combinations of $\widetilde{P}_{j>i+1}$. Finally, by induction we have all $\widetilde{P}_{i<n}$ being positive combination of $\widetilde{P}_n$. However, this contradicts non-replication. To sum up, we proved by contradiction that $\widetilde{P}$ has one unique row. Therefore, $P$ must be identical to $P'$ up to permutations.
\end{proof}

\end{document}